\documentclass[10pt,twocolumn,aps,pra,showpacs,superscriptaddress]{revtex4-1}
\usepackage[english]{babel} 
\usepackage[usenames, dvipsnames]{color} 
\usepackage{graphicx} 
\usepackage{bm} 
\usepackage{amsmath} 
\usepackage{amsthm} 
\usepackage{amssymb} 
\usepackage{times} 
\usepackage{mathrsfs}
\usepackage[pdftex,colorlinks=true,urlcolor=blue,linkcolor=Blue,citecolor=RedViolet]{hyperref} 
\usepackage{appendix}
{\theoremstyle{plain}  \newtheorem{prop}{Proposition}}
{\theoremstyle{plain}  \newtheorem{theo}{Theorem}}
{\theoremstyle{plain}  \newtheorem{corol}{Corollary}}
{\theoremstyle{plain}  }
\newcommand{\hy}[2]{\hyperlink{#1}{\color{black} #2}}
\begin{document}
\title{Determining stationary-state quantum properties directly from system-environment interactions}
\author{F. Nicacio}
\email{nicacio@if.ufrj.br} 
\affiliation{Instituto de F\'isica, Universidade Federal do Rio de Janeiro, 
             21941-972, RJ, Brazil}
\author{M. Paternostro}
\affiliation{School of Mathematics and Physics, 
             Queen's University, Belfast BT7 1NN, UK}
             \author{A. Ferraro}
\affiliation{School of Mathematics and Physics, 
             Queen's University, Belfast BT7 1NN, UK}
%%%%%%%%%%%%%%%%%%%%%%%%%%%%%%%%%%%%%%%%%%%%%%%%%%%%%%%%%%%%%%%%%%%%%%%%%%%%%%%%%%%%%%%%%
\begin{abstract}
\noindent %
Considering stationary states of continuous-variable systems undergoing an open dynamics, 
we unveil the connection between properties and symmetries of the latter and the 
dynamical parameters. 
In particular, we explore the relation between the Lyapunov equation for dynamical 
systems and the steady-state solutions of a time-independent Lindblad master equation 
for bosonic modes. 
Exploiting bona-fide relations that characterize some genuine quantum properties 
(entanglement, classicality, and steerability), we obtain conditions on 
the dynamical parameters for which the system is driven to a steady-state possessing 
such properties.
We also develop a method to capture the symmetries of a steady-state based on 
symmetries of the Lyapunov equation.  
All the results and examples can be useful for steady-state engineering processes. 
\end{abstract}
%%%%%%%%%%%%%%%%%%%%%%%%%%%%%%%%%%%%%%%%%%%%%%%%%%%%%%%%%%%%%%%%%%%%%%%%%%%%%%%%%%%%%%%%%

\maketitle

%%%%%%%%%%%%%%%%%%%%%%%%%%%%%%%%%%%%%%%%%%%%%%%%%%%%%%%%%%%%%%%%%%%%%%%%%%%%%%%%%%%%%%%%%
%%%%%%%%%%%%%%%%%%%%%%%%%%%%%%%%%%%%%%%%%%%%%%%%%%%%%%%%%%%%%%%%%%%%%%%%%%%%%%%%%%%%%%%%%
%%%%%%%%%%%%%%%%%%%%%%%%%       INTRODUCTION     %%%%%%%%%%%%%%%%%%%%%%%%%%%%%%%%%%%%%%%%
%%%%%%%%%%%%%%%%%%%%%%%%%%%%%%%%%%%%%%%%%%%%%%%%%%%%%%%%%%%%%%%%%%%%%%%%%%%%%%%%%%%%%%%%%
%%%%%%%%%%%%%%%%%%%%%%%%%%%%%%%%%%%%%%%%%%%%%%%%%%%%%%%%%%%%%%%%%%%%%%%%%%%%%%%%%%%%%%%%%
The manipulation of the environment affecting the dynamics of a quantum system, 
with the aim of driving the latter towards a specific state, 
embodies a valuable tool for quantum state engineering.
Depending on assumptions about the couplings, the open dynamics can 
lead to either an equilibrium state or to a dynamical steady-state. 
On the other hand, in this scenario, it is critical to ensure that the 
desired state is achieved regardless the fluctuations in the initial state of the system. 
Protocols of this sort are known as reservoir engineering, 
stabilization, and design~\cite{poyatos,wiseman,ticozzi,albert}.  

A standard approach to the modeling of the evolution of an open system 
is the Lindblad master equation (\hypertarget{LME}{LME}) for the density operator 
$\hat \rho$ \cite{wiseman,lindblad,breuer}: 
\begin{equation}                                                                          \label{lindblad}
\frac{d \hat\rho }{d t}  = 
         - \frac{i}{\hbar } [\hat H,  \hat\rho ] 
         - \frac{1}{2\hbar} \! \sum_{m = 1}^{M}
             ( \{ \hat{L}_m^\dag \hat{L}_m , \hat\rho  \}
               - 2 \hat{L}_m \hat\rho \hat{L}_m^\dag      ),
\end{equation} 
which, besides the unitary dynamics ruled by the Hamiltonian operator $\hat H$, accounts 
for a nonunitary dynamics as resulting from the weak 
coupling ({\it via} the operators $\hat L_k$) to uncontrollable environmental degrees 
of freedom. 
The \hy{LME}{LME} is the most general type 
of Markovian and time-homogeneous master equation guaranteeing trace preservation and 
complete positivity. 
Despite the fundamental and very restrictive Markovianity assumption, 
the \hy{LME}{LME} is crucial for the description of an ample set of dynamics in quantum 
optics and information, mesoscopic physics, 
and quantum chemistry~\cite{wiseman,albert2,breuer}.   

In this work we investigate properties and symmetries of continuous-variable states 
driven to equilibrium by a linear evolution 
governed by the time-independent Lindblad dynamics.  
Gaussian states, 
which play a preponderant role in quantum information science and 
are the natural candidates for the implementation of quantum computation with 
continuous variables \cite{lloyd}, belong to such set of states.   

From the mathematical point of view, 
the problem of whether a linear \hy{LME}{LME} has a stable steady-state is equivalent 
to the solution of a Lyapunov equation for the covariance matrix of the quantum state. 
The methodology used to solve Lyapunov equations,   
known as Lyapunov stability theory \cite{dullerud},  
was developed in Ref.~\cite{lyapunov} in the context of dynamical systems.  
This formalism was explored in Ref.~\cite{koga} to determine conditions for a state to be 
pure in the stationary regime. 

In our work, we make use of the connection between the \hy{LME}{LME} and 
the Lyapunov theory to determine several properties of 
continuous-variables steady-states, such as 
classicality \cite{englert}, 
separability \cite{simon2} 
(or bound entanglement \cite{werner}),  
and steerability \cite{wiseman2}.  
Further, we also explore the steady-state symmetries induced by 
the dynamical symmetries of the Lyapunov equation. 
This is particularly interesting, because it is in general hard to characterize 
the symmetries of the states working directly on the \hy{LME}{LME} (\ref{lindblad}). 
This task becomes instead fully manageable when  dealing with finite matrices.   
Our results are applicable to systems with a 
generic number of degrees of freedom and their analyticity brings in turn 
robustness for numerical examinations of the mentioned properties.

The remaining of this paper is organized as follows: 
In Sec.\ref{ldsc}, we set the notation and  describe the linear dynamics, 
discussing the connection between the \hy{LME}{LME} and the Lyapunov theory. 
The mathematical results concerning the Lyapunov equation 
are developed in Sec.\ref{le}, which will be extensively applied to find 
general properties of stationary solutions in Sec.\ref{ss}.
Symmetries of the system are analyzed in Sec.\ref{ess}, while 
examples are given in Sec.\ref{eI}. 
A method for engineering steady-states is presented in Sec.\ref{eps}.  
Section~\ref{conc} presents our conclusions, while in the \hyperlink{Appendix}{Appendixes}
we further discuss some technical aspects of the mathematical approach, including a 
brief summary of the notation.
%
%%%%%%%%%%%%%%%%%%%%%%%%%%%%%%%%%%%%%%%%%%%%%%%%%%%%%%%%%%%%%%%%%%%%%%%%%%%%%%%%%%%%%%%%%
%%%%%%%%%%%%%%%%%%%%%%%%%%%%%%%%%%%%%%%%%%%%%%%%%%%%%%%%%%%%%%%%%%%%%%%%%%%%%%%%%%%%%%%%%
\section{Linear Dynamics and Stationary Conditions}\label{ldsc} %%%%%%%%%%%%%%%%%%%%%%%%%
%%%%%%%%%%%%%%%%%%%%%%%%%%%%%%%%%%%%%%%%%%%%%%%%%%%%%%%%%%%%%%%%%%%%%%%%%%%%%%%%%%%%%%%%%
%%%%%%%%%%%%%%%%%%%%%%%%%%%%%%%%%%%%%%%%%%%%%%%%%%%%%%%%%%%%%%%%%%%%%%%%%%%%%%%%%%%%%%%%%
%
For a system of $n$ continuous degrees of freedom (\hypertarget{DF}{DF}), 
the generalized coordinates together with the canonical conjugated momenta are collected 
in a $2n$ column vector:
\begin{equation}                                                                         \label{vecx}
\hat x := (\hat q_1,...,\hat q_n, \hat p_1,...\hat p_n)^\dag.
\end{equation}
In this notation, the canonical commutation relation (\hypertarget{CCR}{CCR}) 
is written compactly as $[\hat x_j , \hat x_k ] = i \hbar \, \mathsf J_{jk}$ 
with ${\sf J}_{jk}$ given by the elements of the symplectic matrix 
\begin{equation}                                                                         \label{comm}
 \mathsf J := 
\left(  
\begin{array}{cc} 
       {\bf 0}_n   & \mathsf I_n  \\
      -\mathsf I_n & {\bf 0}_n 
       \end{array}
 \right).
\end{equation}

We consider the evolution of a quantum state governed by the \hy{LME}{LME} with  
a quadratic Hamiltonian and linear Lindblad operators, {\it viz.},   
\begin{equation}                                                                         \label{ham-lind}
\begin{aligned}
&\hat H = \frac{1}{2} \hat x \cdot \mathbf H \hat x + 
\xi \cdot \mathsf J \hat x + H_0,\\
&\hat L_m = \lambda_m  \cdot \mathsf J \hat x + \mu_m,
\end{aligned}
\end{equation}
where $\xi \in \mathbb R^{2n}$ and $\lambda_m \in \mathbb C^{2n}$ are column vectors; 
$H_0 \in \mathbb R$ and $\mu_m \in \mathbb C$ are constants and $m = 1,...,M$.
The Hessian of the Hamiltonian is symmetric by definition: 
$\mathbf H = \mathbf H^\top \in {\rm Mat}(2n,\mathbb R)$.
Under such conditions, the evolution of the mean value vector 
$\langle \hat x \rangle_t  :=  {\rm Tr}\left[ \hat x \hat\rho (t) \right]
                               \in \mathbb R^{2n}$ 
can be retrieved from (\ref{lindblad}) by using only the 
\hy{CCR}{CCR} \cite{nicacio,wiseman}:
\begin{equation}                                                                         \label{mveq}
\frac{d\langle \hat x \rangle_t}{d t}  = 
\xi - \eta + { \bf \Gamma } \langle \hat x \rangle_t ,
\end{equation}
where we have introduced 
$\eta := \sum_{m = 1}^M {\rm Im}(\mu_m^\ast \, \lambda_m) \in {\mathbb C}^{2n}$, 
\begin{equation}                                                                         \label{gammamat}
{\bf \Gamma }  := \mathsf J \mathbf H - {\rm Im} {\bf \Upsilon} \mathsf J 
                 \in {\rm Mat}(2n,\mathbb R),   
\end{equation}
and 
\begin{equation}                                                                         \label{upsilon}
{\bm \Upsilon} := 
\sum_{m = 1}^M \lambda_m \lambda_m^\dagger \in {\rm Mat}(2n,\mathbb C). 
\end{equation}
 
The natural question that arises at this point is whether 
a solution of (\ref{mveq}) attains a finite asymptotic value when $t \to \infty$.  
The answer is provided in the context of the Lyapunov stability theory \cite{dullerud}. 
All solutions will be driven to an asymptotic point, 
if the matrix $\bf \Gamma$ is {\it asymptotically stable} (\hypertarget{AS}{AS}), 
{\it i.e.}, if its spectrum has negative real part. 

Interestingly enough, from the Lindblad dynamics (\ref{lindblad}) 
with the operators in Eq.(\ref{ham-lind}), 
a Lyapunov equation (\hypertarget{LE}{LE}) emerges naturally for the stationary value of 
the covariance matrix (\hypertarget{CM}{CM}) of the system, as we shall see. 
Defining the \hy{CM}{CM} of the state $\hat \rho$ as  
${\bf V} = {\bf V}^\top \in {\rm Mat}(2n,\mathbb R) $ with elements \cite{note1}
\begin{equation}                                                                         \label{cmdef}
\mathbf V_{\! jk} =  
\tfrac{1}{\hbar} {\rm Tr}\left[ 
                         \left\{ \hat x_j - \langle \hat x_j \rangle_t , 
                         \hat x_k - \langle \hat x_k \rangle_t \right\}
                         \hat\rho(t)
                     \right],  
\end{equation}
and calculating its evolution \cite{nicacio,wiseman}, 
the (possible) stationary value of the \hy{CM}{CM} will be the solution of the 
\hy{LE}{LE} 
\begin{equation}                                                                         \label{lyapeq2}
{\bf V}{\bf \Gamma}^{\! \top} + {\bf \Gamma} {\bf V} + {\bf D} = 0, 
\end{equation}
with $\bf \Gamma$ in (\ref{gammamat}) and  
\begin{equation}                                                                         \label{dmat}
{\bf D} :=  2 \, {\rm Re}{\bf \Upsilon} = {\bf D}^{\!\top} \in {\rm Mat}(2n,\mathbb R),    
\end{equation}
which is positive semidefinite, ${\bf D} \ge 0$, by the definition of $\bf \Upsilon$.   
The Lyapunov theorem and its extensions \cite{lyapunov,horn,dullerud}
guarantee that for an \hy{AS}{AS} matrix $\bf \Gamma$, 
the solution of Eq.(\ref{lyapeq2}) exists and is unique. 
Furthermore, those theorems also relate the stability nature of the matrix $\bf \Gamma$ 
%and consequently the existence of an asymptotic solution for (\ref{mveq}), 
to the existence of matrices (in our case $\bf V$ and $\bf D$) 
satisfying the \hy{LE}{LE} (\ref{lyapeq2}). 

We stress that, in order to deduce Eqs.(\ref{mveq}) and (\ref{lyapeq2}),
we do not need any assumptions about the initial state of the system. 
The derivation of such equations only uses the \hy{LME}{LME}, 
the particular structure of Eqs.~(\ref{ham-lind}), and the \hy{CCR}{CCR}.
Meanwhile, the \hy{LME}{LME} with the operators (\ref{ham-lind}) will always preserve the 
Gaussian character of an initial Gaussian state \cite{nicacio}. 
Once the \hy{CM}{CM} of a steady-state of the system is a solution of (\ref{lyapeq2}), 
which is unique and does not depend on the initial state, 
any initial state will end in a Gaussian steady-state. 

%%%%%%%%%%%%%%%%%%%%%%%%%%%%%%%%%%%%%%%%%%%%%%%%%%%%%%%%%%%%%%%%%%%%%%%%%%%%%%%%%%%%%%%%%
%%%%%%%%%%%%%%%%%%%%%%%%%%%%%%%%%%%%%%%%%%%%%%%%%%%%%%%%%%%%%%%%%%%%%%%%%%%%%%%%%%%%%%%%%
\section{Lyapunov Equations}\label{le}                          %%%%%%%%%%%%%%%%%%%%%%%%%
%%%%%%%%%%%%%%%%%%%%%%%%%%%%%%%%%%%%%%%%%%%%%%%%%%%%%%%%%%%%%%%%%%%%%%%%%%%%%%%%%%%%%%%%%
%%%%%%%%%%%%%%%%%%%%%%%%%%%%%%%%%%%%%%%%%%%%%%%%%%%%%%%%%%%%%%%%%%%%%%%%%%%%%%%%%%%%%%%%%
In this section we show and develop results concerning the generic Lyapunov equation
\begin{equation}                                                                         \label{lyapeq}
{\bf A}\, {\bf P} + 
{\bf P} {\bf A}^{\!\dagger} + {\bf Q} = 0,   
\end{equation}
and its solution. 
Since our objective is to understand properties of stationary solutions 
of the \hy{LME}{LME}, it is convenient to assume that 
({\it i})   ${\bf A} \in {\rm Mat}(m,\mathbb C)$ is \hy{AS}{AS},    
({\it ii})  ${\bf P} = {\bf P}^\dagger \in {\rm Mat}(m,\mathbb C)$ and 
({\it iii}) ${\bf Q} = {\bf Q}^\dagger \in {\rm Mat}(m,\mathbb C)$. 
From now on, the \hy{LE}{LE} (\ref{lyapeq}) will be represented by the triple 
$\lfloor{\bf P},{\bf A},{\bf Q}\rceil$.   

The first of those assumptions (${\bf A}$ is \hy{AS}{AS}) 
is enough to prove \cite{horn,dullerud} 
that the {\it unique} solution for the \hy{LE}{LE} in (\ref{lyapeq}) is written as 
\begin{equation}                                                                         \label{sol}
{\bf P}({\bf Q},{\bf A}) = \int_{0}^{\infty} \!\!\, dt 
             \,\, 
             {\rm e }^{{\bf A} t } \, 
             {\bf Q}  \,\, 
             {\rm e }^{{\bf A}^{\!\dagger} t } \, . 
\end{equation}
Furthermore, for any  
${\bf A} \in {\rm GL}(m,\mathbb C)$, it is true that
\begin{equation}
{\rm In} \left({\rm e }^{{\bf A}  t } \, 
{\bf Q}  \,\, 
{\rm e }^{{\bf A}^{\!\dagger} t }\right) = 
{\rm In}\left( {\bf Q}  \right), 
\end{equation}
because the expression inside the parenthesis on the LHS is a 
congruence transformation of $\bf Q$. 
The symbol ``$\rm In$'' refers to the {\it inertia index} of a matrix, 
as defined in \hyperlink{Appendix}{Appendix I}.   
On the other hand, once ${\bf A}$ is \hy{AS}{AS}, then 
$\lim_{t \to \infty} {\rm e }^{{\bf A} t } = 0$, 
which guarantees the convergence of the integral in (\ref{sol}). 
These last arguments about the structure of Eq.(\ref{sol}) 
are used in the proof of the following result \cite{dullerud,horn}:
%
%%%%%%%%%%%%%%%%%%%%%%%%%%%%%%%%%%%%%%%%%%%%%%%%%%%%%%%
%%%%%%%%%%%%%%%%%%%%%%%%%%%%%%%%%%%%%%%%%%%%%%%%%%%%%%%
\begin{prop}                                                                            \label{prop:inert1}
Consider the solution (\ref{sol}) 
for the \hy{LE}{LE} in (\ref{lyapeq}). 
If ${\bf Q} \ge 0$ (resp. $ {\bf Q} > 0$),  then 
${\bf P} \ge 0$ (resp. ${\bf P} > 0$) .
\end{prop} %
%%%%%%%%%%%%%%%%%%%%%%%%%%%%%%%%%%%%%%%%%%%%%%%%%%%%%%%
%%%%%%%%%%%%%%%%%%%%%%%%%%%%%%%%%%%%%%%%%%%%%%%%%%%%%%%
%
\noindent Note that Prop.\ref{prop:inert1} does not exclude the statement  
${\bf Q} \ge 0 \Longrightarrow {\bf P} > 0$, since the set of matrices ${\bf P}$ 
such that $ {\bf P} > 0 $ is a subset of $ {\bf P} \ge 0 $, 
{\it cf.} \hyperlink{Appendix}{Appendix I}.  
This is the case provided the pair $({\bf Q},{\bf A})$ 
is observable \cite{horn,dullerud}. 
For our purposes, this statement is not necessary, 
however it is for the results in \cite{koga}. 

Now, let us go a bit further with the results in Prop.\ref{prop:inert1},  
specializing the properties of the \hy{AS}{AS} matrix $\bf A$:  
%%%%%%%%%%%%%%%%%%%%%%%%%%%%%%%%%%%%%%%%%%%%%%%%%%%%%%%
%%%%%%%%%%%%%%%%%%%%%%%%%%%%%%%%%%%%%%%%%%%%%%%%%%%%%%%
\begin{prop}                                                                             \label{prop:inert2}
Consider the \hy{LE}{LE} 
$\lfloor{\bf P},{\bf A},{\bf Q}\rceil$ with
${\bf A} = {\bf A}^{\!\dagger}$, 
then 
${\bf P} \ge 0$ (resp. ${\bf P} > 0$)  
if and only if 
${\bf Q} \ge 0$ (resp. ${\bf Q} > 0$).
\end{prop} 
%%%%%%%%%%%%%%%%%%%%%%%%%%%%%%%%%%%%%%%%%%%%%%%%%%%%%%%
\begin{proof}[Proof]
Since ${\bf A}$ is self-adjoint and negative 
definite (\hy{AS}{AS}), it is possible to write  
$ {\bf A} = - \sqrt{-{\bf A}} 
              \sqrt{-{\bf A}}$, 
where $\sqrt{-{\bf A}}$ is the unique self-adjoint 
positive definite square root of $-{\bf A}$.
From the \hy{LE}{LE}~(\ref{lyapeq}), 
\begin{eqnarray}
 \!\!\!\!\! {\rm Spec}_{\mathbb R}( {\bf Q}) & = & 
- {\rm Spec}_{\mathbb R}
  \left( 
        {\bf A}{\bf P} + 
        {\bf P} {\bf A} 
\right)                                      \\
&=& - {\rm Spec}_{\mathbb R} 
\left[
      \left( {\bf P} + 
            {\bf A} {\bf P} {\bf A}^{-1}
     \right) {\bf A}
\right]                                      \nonumber \\
& = & {\rm Spec}_{\mathbb R}
\left[ \sqrt{-{\bf A}}
       \left( {\bf P} + 
            {\bf A} {\bf P} {\bf A}^{-1}
     \right)
     \sqrt{-{\bf A}}
\right]\!.                                   \nonumber
\end{eqnarray}
Since the sum of positive semidefinite 
(resp. positive definite) 
matrices is positive semidefinite 
(resp. positive definite), 
and since a congruence transformation does not change 
the signs of the eigenvalues 
[or the inertia of the matrix], 
it follows that 
${\bf Q} \ge 0$ (resp. ${\bf Q} > 0$) if 
$ {\bf P} \ge 0$ (resp. ${\bf P} > 0$), 
which proves the necessary condition. 
The sufficiency is in Prop.\ref{prop:inert1}.  
\end{proof}
%%%%%%%%%%%%%%%%%%%%%%%%%%%%%%%%%%%%%%%%%%%%%%%%%%%%%%%
%
\noindent Note that, once one statement in Prop.\ref{prop:inert2} is 
${\bf Q} > 0 \Longleftrightarrow {\bf P} > 0$, then 
the statement 
${\bf Q} \ge 0 \Longleftrightarrow {\bf P} \ge 0$ necessarily means that 
if $\bf Q$ has one null eigenvalue, then $\bf P$ will also have, 
and {\it vice-versa}.  

In the direction of the main task of this work, 
we must develop some results concerning matrices of the form ${\bf P} + {\bf \Xi}$, 
where ${\bf P}$ is the solution in Eq.(\ref{sol}) of 
$\lfloor{\bf P},{\bf A},{\bf Q}\rceil$ and 
${\bf \Xi} = {\bf \Xi}^\dagger \in {\rm Mat}(m,\mathbb C)$.  
%%%%%%%%%%%%%%%%%%%%%%%%%%%%%%%%%%%%%%%%%%%%%%%%%%%%%%%
%%%%%%%%%%%%%%%%%%%%%%%%%%%%%%%%%%%%%%%%%%%%%%%%%%%%%%%
\begin{corol}                                                                              \label{cor:if} 
${\bf P} + {\bf \Xi} \ge 0 $ \, if \, 
${\bf Q}_{[\bf \Xi]} := 
{\bf Q} - {\bf \Xi A}^\dagger - {\bf A \Xi} \ge 0 $.
\end{corol}
%%%%%%%%%%%%%%%%%%%%%%%%%%%%%%%%%%%%%%%%%%%%%%%%%%%%%%%
\begin{proof}[Proof] 
It is easy to see that the \hy{LE}{LE} 
$\lfloor {\bf P} + {\bf \Xi},
         {\bf A},
         {\bf Q}_{[\bf \Xi]}   \rceil$    
is equivalent to the \hy{LE}{LE} in (\ref{lyapeq}). 
Thus, the proof follows from 
Prop.\ref{prop:inert1}, since 
${\bf Q}_{[\bf \Xi]}={\bf Q}_{[\bf \Xi]}^\dagger$. 
\end{proof}
%%%%%%%%%%%%%%%%%%%%%%%%%%%%%%%%%%%%%%%%%%%%%%%%%%%%%%%
%%%%%%%%%%%%%%%%%%%%%%%%%%%%%%%%%%%%%%%%%%%%%%%%%%%%%%%
The converse statement of Corol.\ref{cor:if} is not true in general.
By using the restriction for the matrices $\bf A$, 
as in Prop.\ref{prop:inert2}, we obtain the following corollary giving a 
necessary and sufficient condition. 
%%%%%%%%%%%%%%%%%%%%%%%%%%%%%%%%%%%%%%%%%%%%%%%%%%%%%%%
%%%%%%%%%%%%%%%%%%%%%%%%%%%%%%%%%%%%%%%%%%%%%%%%%%%%%%%
\begin{corol}                                                                            \label{cor:iff}
Consider ${\bf A} = {\bf A}^{\!\dagger}$ ($\bf A$ 
is \hy{AS}{AS}), 
then  ${\bf P} +  {\bf \Xi}  \ge 0$ \,
if and only if \, 
$\widetilde{\bf Q}_{[{\bf \Xi}]} := {\bf Q}
- {\pmb \{} {\bf \Xi, A} {\pmb \}}_{\!+} \ge 0 $.
\end{corol}
%%%%%%%%%%%%%%%%%%%%%%%%%%%%%%%%%%%%%%%%%%%%%%%%%%%%%%%
\begin{proof}[Proof]
As before, we construct the equivalent 
Lyapunov equation
$\lfloor {\bf P} + {\bf \Xi},
         {\bf A},
         \widetilde{\bf Q}_{[\bf \Xi]}   \rceil$,    
and use Prop.\ref{prop:inert2} since
$\widetilde{\bf Q}_{[\bf \Xi]}\!=\! 
 \widetilde{\bf Q}_{[\bf \Xi]}^\dagger$. 
\end{proof}
%%%%%%%%%%%%%%%%%%%%%%%%%%%%%%%%%%%%%%%%%%%%%%%%%%%%%%%
%%%%%%%%%%%%%%%%%%%%%%%%%%%%%%%%%%%%%%%%%%%%%%%%%%%%%%%

\noindent Note that, contrary to Props.~\ref{prop:inert1} and~\ref{prop:inert2}, 
we did not mention the strictly positive cases 
(${\bf Q}_{[\bf \Xi]}\!>\!0$, $\widetilde{\bf Q}_{[\bf \Xi]}\!>\!0$) in 
Corols.~{\ref{cor:if}} and~\ref{cor:iff}. 
Actually, these cases follow the same prescription, 
but they are not necessary for our next results. 

In what follows, properties of the steady states, driven to equilibrium under 
the linear evolution generated by (\ref{ham-lind}), will be considered from the 
perspective of the results developed in this section. 
%%%%%%%%%%%%%%%%%%%%%%%%%%%%%%%%%%%%%%%%%%%%%%%%%%%%%%%%%%%%%%%%%%%%%%%%%%%%%%%%%%%%%%%%%
%%%%%%%%%%%%%%%%%%%%%%%%%%%%%%%%%%%%%%%%%%%%%%%%%%%%%%%%%%%%%%%%%%%%%%%%%%%%%%%%%%%%%%%%%
\section{ Bona-Fide Relations and Steady-States }\label{ss}     %%%%%%%%%%%%%%%%%%%%%%%%%
%%%%%%%%%%%%%%%%%%%%%%%%%%%%%%%%%%%%%%%%%%%%%%%%%%%%%%%%%%%%%%%%%%%%%%%%%%%%%%%%%%%%%%%%%
%%%%%%%%%%%%%%%%%%%%%%%%%%%%%%%%%%%%%%%%%%%%%%%%%%%%%%%%%%%%%%%%%%%%%%%%%%%%%%%%%%%%%%%%%
Through the temporal evolution of a state by the \hy{LME}{LME} conditioned 
to an \hy{AS}{AS} dynamics, the dependence on the initial condition is progressively 
erased by the environmental action. 
Therefore the steady-state properties must be completely determined only by 
the environment.  
An usual way to describe properties of continuous-variable states is given by 
\textit{bona-fide} relations involving the \hy{CM}{CM} of the states.
From now on, we will assume that the \hy{CM}{CM} of a quantum state evolves with 
$\bf \Gamma $ \hy{AS}{AS} and ${\bf D} \ge 0$ and attains an asymptotic value described 
by the \hy{LE}{LE} (\ref{lyapeq2}) with solution 
${\bf V} := {\bf P}({\bf D},{\bf \Gamma})$ in Eq.(\ref{sol}).

It is convenient to recall the definitions of the auxiliary matrices defined in 
Corols.~{\ref{cor:if}} and~\ref{cor:iff}, but now for the LE in question. 
For any matrix ${\bf \Xi} = {\bf \Xi}^\dagger \in {\rm Mat}(2n,\mathbb C)$, we have  
\begin{equation}
{\bf D}_{[\bf \Xi]} :=  {\bf D} - {\bf \Xi \, \Gamma}{\!^\top} - {\bf \Gamma \, \Xi}, \,\,\,
\widetilde{\bf D}_{[{\bf \Xi}]} := {\bf D} - {\pmb \{} {\bf \Xi,\Gamma} {\pmb \}}_{\!+}. 
\end{equation}
%

%%%%%%%%%%%%%%%%%%%%%%%%%%%%%%%%%%%%%%%%%%%%%%%%%%%%%%%%%%%%%%%%%%%%%%%%%%%%%%%%%%%%%%%%%
\subsection{Uncertainty Principle}\label{uc}                 %%%%%%%%%%%%%%%%%%%%%%%%%%%%
%%%%%%%%%%%%%%%%%%%%%%%%%%%%%%%%%%%%%%%%%%%%%%%%%%%%%%%%%%%%%%%%%%%%%%%%%%%%%%%%%%%%%%%%%
Any quantum state is subjected to constrains imposed by the uncertainty principle, 
which is only a consequences of the \hy{CCR}{CCR}. 
For the continuous-variables case, 
this principle takes into account only the \hy{CM}{CM} (\ref{cmdef}). 
A genuine physical state has a \hy{CM}{CM} such that \cite{simon}
\begin{equation}                                                                         \label{up1}
{\bf V} + i \mathsf J_{2n} \ge 0. 
\end{equation}

Given a Hamiltonian and a collection of Lindblad operators 
as in (\ref{ham-lind}), 
what can our corollaries say about the genuineness of the steady-state 
generated by the \hy{LME}{LME}? 
Invoking Corol.\ref{cor:if}, the matrix $\bf V$ of a steady-state is 
a bona-fide \hy{CM}{CM} if 
${\bf D}_{ [i \mathsf J] } = 
{\bf D} - i {\bf \Gamma} \mathsf J
        - i \mathsf J {\bf \Gamma}^{\!\top} \ge 0 $.
However, using Eqs.~(\ref{gammamat}) and~(\ref{upsilon}),  
it is not difficult to show that ${\bf D}_{[i \mathsf J]} = 2 {\bf \Upsilon}^\ast $,
which is always positive semidefinite, according to the definition of 
${\bf \Upsilon}$ in (\ref{upsilon}). 
Tautologically, this says that all linear \hy{LME}{LMEs} with 
$\bf \Gamma$ \hy{AS}{AS} and ${\bf D} \ge 0$ 
will drive the system to a steady-state obeying the relation (\ref{up1}), {\it i.e.}, 
a genuine physical state.

On the other hand, Eq.(\ref{lyapeq2}) is a consequence of the \hy{CCR}{CCR}, 
as mentioned before. 
Accordingly, the \hy{LME}{LME} guarantees that the uncertainty principle holds 
for all times, including the steady-state limit, 
whereof the condition in Corol.{\ref{cor:if}} is necessary and sufficient 
regardless of whether $\bf \Gamma$ is symmetric.  
Before going to the next bona-fide relation, it is important to remark that this 
extension of Corol.{\ref{cor:if}} to an ``if and only if'' condition is only true 
for the relation in Eq.(\ref{up1}). 
For all the other relations which will appear in what follows, the differences between 
Corols.~\ref{cor:if} and~\ref{cor:iff} should be considered. 

\

%%%%%%%%%%%%%%%%%%%%%%%%%%%%%%%%%%%%%%%%%%%%%%%%%%%%%%%%%%%%%%%%%%%%%%%%%%%%%%%%%%%%%%%%%
\subsection{Classical States}\label{cs}    %%%%%%%%%%%%%%%%%%%%%%%%%%%%%%%%%%%%%%%%%%%%%%
%%%%%%%%%%%%%%%%%%%%%%%%%%%%%%%%%%%%%%%%%%%%%%%%%%%%%%%%%%%%%%%%%%%%%%%%%%%%%%%%%%%%%%%%%
Classical States are defined as having a positive Glauber-Sudarshan distribution 
function, they are also called P-re\-pre\-sen\-table states. 
This definition relies on the possibility to express a desired state as a classical 
mixture of coherent states.
The necessary and sufficient condition for P-representability  of a Gaussian state 
is written in terms of its \hy{CM}{CM} as the bona-fide relation 
(see \hyperlink{Appendix}{Appendix II})
\begin{equation}                                                                         \label{bfclas}
{\bf V} - \mathsf I_{2n} \ge 0. 
\end{equation}

The evolution that drives the system to a classical steady-state is subjected to the 
sufficient condition given by Corol.\ref{cor:if}:
\begin{equation}                                                                         \label{clas}
{\bf D}_{[-\mathsf I]}=  {\bf \Gamma} + 
                       {\bf \Gamma}^{\!\top} + {\bf D} \ge 0 
\Longrightarrow {\bf V} - \mathsf I_{2n} \ge 0. 
\end{equation}
The contrapositive of the above statement says that 
if a given \hy{LME}{LME} is such that ${\bf D}_{[-\mathsf I]}$ 
has at least one negative eigenvalue, 
it will lead the system to a nonclassical stationary state. 
The matrix ${\bf V}$ is, by hypothesis, the \hy{CM}{CM} of a steady-state of 
the \hy{LME}{LME}. 
Once the converse statement of (\ref{clas}) is not true, one can conclude that
there are classical steady-states which can not be generated by a \hy{LME}{LME} 
with dynamical matrices such that ${\bf D}_{[-\mathsf I]} \ge 0$. 

If we consider only steady-states generated by a \hy{LME}{LME} with ${\bf \Gamma}$ 
symmetric, by Corol.\ref{cor:iff} the possible classical states will obey the 
necessary and sufficient condition
\begin{equation}                                                                         \label{clas2}
\widetilde{\bf D}_{[-\mathsf I]} = 2{\bf \Gamma} + {\bf D} \ge 0 
\Longleftrightarrow {\bf V} - \mathsf I_{2n} \ge 0.  
\end{equation}
This means that all states generated by a \hy{LME}{LME} with 
${\bf \Gamma} = {\bf \Gamma}^{\!\top}$ and $\widetilde{\bf D}_{[-\mathsf I]} \ge 0$ 
are classical states. 
Conversely, all classical states with \hy{CM}{CM} $\bf V$ which are solutions of 
a \hy{LE}{LE} with ${\bf \Gamma} = {\bf \Gamma}^{\!\top}$ are steady-states of a 
\hy{LME}{LME} satisfying $\widetilde{\bf D}_{[-\mathsf I]} \ge 0$. 

Since classicality is related to mixtures of coherent states,  
one instructive example is given by the \hy{CM}{CM} 
${\bf V} = \mathsf I_{2n}$ --- 
{\it i.e.}, the \hy{CM}{CM} of any $n$-mode (or $n$-\hy{DF}{DF}) coherent state.
The simplicity of this case enables us to derive an useful 
necessary and sufficient condition besides the relation in Eq.(\ref{clas2}). 
Actually, we will be concerned with the slightly more general situation: 
the tensor product of $n$ Gibbs-states with the same occupation number. 
These states have the global \hy{CM}{CM} written as ${\bf V} = \alpha \mathsf I_{2n}$, 
which is a \hy{CM}{CM} of a classical state if $\alpha \ge 1$. 
%%%%%%%%%%%%%%%%%%%%%%%%%%%%%%%%%%%%%%%%%%%%%%%%%%%%%%%
%%%%%%%%%%%%%%%%%%%%%%%%%%%%%%%%%%%%%%%%%%%%%%%%%%%%%%%
\begin{corol}                                                                             \label{corol:diag} 
For any $\alpha > 0$ and any $\bf \Gamma$ such that 
$({\bf \Gamma} + {\bf \Gamma}^{\!\top}) \le 0$, 
the matrix ${\bf V} = \alpha \, \mathsf I_{2n}$ 
is a solution of the \hy{LE}{LE} 
$\lfloor{\bf V},{\bf \Gamma},{\bf D}\rceil$ 
if and only if 
$ \alpha ( {\bf \Gamma} + {\bf \Gamma}^{\!\top} ) 
+ {\bf D} = 0$.
Furthermore, 
$\alpha = \tfrac{1}{2}{\rm Tr} {\bf D}/
          {\rm Tr}({\rm Im} {\bf \Upsilon} \mathsf J)$.
% The matrix ${\bf V} = \alpha \, \mathsf I_{2n}$ 
% is a solution of the \hy{LE}{LE} 
% $\lfloor{\bf V},{\bf \Gamma},{\bf D}\rceil$ 
% if and only if 
% $ \, \alpha ( {\bf \Gamma} + {\bf \Gamma}^{\!\top} ) 
% + {\bf D} = 0 \, , \forall \alpha > 0 $ and 
% $\forall \bf \Gamma$ such that 
% $({\bf \Gamma} + {\bf \Gamma}^{\!\top}) \le 0$. 
% %
% Furthermore, 
% $\alpha = - \tfrac{1}{2}{\rm Tr} {\bf D}/
%                       {\rm Tr}\, {\bf \Gamma} $.
\end{corol}
%%%%%%%%%%%%%%%%%%%%%%%%%%%%%%%%%%%%%%%%%%%%%%%%%%%%%%%
\begin{proof}[Proof]
The sufficient condition is trivially obtained by
constructing the \hy{LE}{LE} 
$\lfloor \alpha \mathsf I_{2n},
         {\bf \Gamma}, {\bf D} \rceil$
from (\ref{lyapeq2}). 
To prove the necessary condition, one defines 
\begin{equation*}
\mathcal I := \int^\infty_0 \!\!\! dt\,  
             {\rm e }^{{\bf \Gamma} t } \,  
             {\rm e }^{{\bf \Gamma}^{\!\top}\!t } 
\end{equation*}
and integrates it by parts to show that 
$\mathcal I\, {\bf \Gamma}^{\!\top} \!+\!  
{\bf \Gamma} \, \mathcal I = 
-\mathsf{I}_{2n}$.
The solution (\ref{sol}) with 
${\bf Q} = {\bf D} = -\alpha ({\bf \Gamma} + 
                          {\bf \Gamma}^{\!\top})$ 
and ${\bf A}={\bf \Gamma}$  
shows that 
${\bf V} = -\alpha ({\bf \Gamma} \, \mathcal I +
           \mathcal I \, {\bf \Gamma}^{\!\top}) = 
           \alpha \mathsf I$. 
Since ${\bf \Gamma} \in {\rm Mat}(2n,\mathbb R)$, 
its complex eigenvalues will occur in 
conjugate pairs, 
then ${\rm Tr} \, {\bf \Gamma} \in \mathbb R$. 
Since it is also \hy{AS}{AS}, 
${\rm Tr}\,{\bf \Gamma} \ne 0$, 
then the value of $\alpha$ holds if one considers 
the definition of $\bf \Gamma$ in Eq.(\ref{gammamat}) 
and the fact that ${\rm Tr}(\mathsf J {\bf H}) = 0$, 
since $\bf H$ is symmetric. 
\end{proof}
%%%%%%%%%%%%%%%%%%%%%%%%%%%%%%%%%%%%%%%%%%%%%%%%%%%%%%%
%%%%%%%%%%%%%%%%%%%%%%%%%%%%%%%%%%%%%%%%%%%%%%%%%%%%%%%

The state studied in Corol.\ref{corol:diag} is an example of the multiplicity of the 
steady-state with respect to different matrices $\bf \Gamma$ and $\bf D$.
There is an infinite number of matrices satisfying the relation 
$\alpha ( {\bf \Gamma} + {\bf \Gamma}^{\!\top} ) + {\bf D} = 0$ and  
giving rise to the same steady state. However, for a given pair of matrices 
$\bf \Gamma$ and $\bf D$, the solution ${\bf V(\Gamma, D)}$ 
is uniquely given in Eq.(\ref{sol}).  

%%%%%%%%%%%%%%%%%%%%%%%%%%%%%%%%%%%%%%%%%%%%%%%%%%%%%%%%%%%%%%%%%%%%%%%%%%%%%%%%%%%%%%%%%
\subsection{ Separable States }\label{seps}              %%%%%%%%%%%%%%%%%%%%%%%%%%%%%%%%              
%%%%%%%%%%%%%%%%%%%%%%%%%%%%%%%%%%%%%%%%%%%%%%%%%%%%%%%%%%%%%%%%%%%%%%%%%%%%%%%%%%%%%%%%%
A necessary and sufficient condition for an 
$n$-mode Gaussian state to be {\it separable} with respect to one of the modes, 
say the $k^{\underline{\text{th}}}$ one, is defined in terms of its \hy{CM}{CM} as 
${\bf V}^{\!\top_{\!\! k}} + i \mathsf J \ge 0$, where  
$ {\bf V}^{\!\top_{\!\! k}} := {\bf T}_{\!k} {\bf V} {\bf T}_{\!k} $. 
The transformation 
${\bf T}_{\!k} = {\bf T}_{\!k}^{-1} = {\bf T}_{\!k}^\top$ 
is a local time inversion on the $\hat x$ operator, {\it viz}.,  
\begin{equation}
{\bf T}_{\!k} \hat x = (\hat q_1,..., \hat q_k,...,q_n,
                        \hat p_1,...,-\hat p_k,...,\hat p_n)^\top
\end{equation}
and, of course, can not be implemented unitarily. 
Since ${\bf T}_{\!k}$ is orthogonal, we can express the separability condition
equivalently as \cite{simon2,werner}
\begin{equation}                                                                         \label{bfsep}
{\bf V} + i \mathsf J^{\!\top_{\!\! k}} \ge 0.  
\end{equation}

The statements in Eqs.~(\ref{clas}) and~(\ref{clas2}) 
can be readily modified to the present case:  
\begin{subequations}                                                                     \label{seprel}
\begin{eqnarray}                                                                                
(\text{Corollary~\ref{cor:if}} ) \,\,\, 
{\bf D}_{[i \mathsf J^{\!\top_{\!\! k}}]}  \ge 0 
&\Longrightarrow& {\bf V} + i \mathsf J^{\!\top_{\!\! k}} \ge 0;                         \label{seprel:1} \\
(\text{Corollary~\ref{cor:iff}} ) \,\,\, 
\widetilde{\bf D}_{[i \mathsf J^{\!\top_{\!\! k}}]}  \ge 0 
&\Longleftrightarrow& {\bf V} + i \mathsf J^{\!\top_{\!\! k}} \ge 0.                     \label{seprel:2}  
\end{eqnarray}
\end{subequations}
The interpretations of these conditions are also readily adap\-ted 
from those in the previous subsection, it is just a question of changing 
the dichotomy ``classical/nonclassical'' to ``separable/entangled''.  

All classical states (not only the Gaussian ones) are separable, since they are written 
as a convex sum of coherent states which are separable [see Eq.(\ref{prep})]. 
As a consequence, a hierarchy of the dynamics of \hy{LME}{LME}s can be established: 
the set of matrices such that 
$\widetilde{\bf D}_{[-\mathsf I]} \ge 0$ in (\ref{clas2}) is a subset of those 
satisfying 
$\widetilde{\bf D}_{[i \mathsf J^{\!\top_{\!\! k}}]} \ge 0$ in (\ref{seprel:2}).
However, this is not true for the matrices in (\ref{clas}) and (\ref{seprel:1}), 
because both only give a sufficient condition. 

Now, consider the following partition of the number of \hy{DF}{DF} 
of a state: $ n = n_1 + n_2$, where $n_i$ is the number of 
\hy{DF}{DF} of each partition. Define also the local time inversion operation as 
\begin{equation}                                                                         \label{pt}
\!\!{\bf T}_{\!n_2} \hat x =\! (\hat q_1,...,\hat q_n,\hat p_1,...,\hat p_{n_1},
                         -\hat p_{n_1+1},...,-\hat p_{n_1+n_2})^\top\! .
\end{equation}
The separability criteria already exposed is a necessary and sufficient condition 
only if $n_2 = 1$. For all other cases, entangled states with 
${\bf V} + i {\bf T}_{\!n_2} \mathsf J {\bf T}_{\!n_2} \ge 0$ 
are bound entangled, {\it i.e.}, 
they have nondistillable entanglement \cite{werner}. 
One can also relate the reservoir properties with this bona-fide relation
through the replacement  
$\mathsf J^{\!\top_{\!\! k}} \to {\bf T}_{\!n_2} \mathsf J {\bf T}_{\!n_2}$ 
in Eqs.(\ref{seprel}).  
Note that the bona-fide relation in question does not say whether the state is separable 
or bound-entangled. 

%%%%%%%%%%%%%%%%%%%%%%%%%%%%%%%%%%%%%%%%%%%%%%%%%%%%%%%%%%%%%%%%%%%%%%%%%%%%%%%%%%%%%%%%%
\subsection{ Gaussian Steerability }      %%%%%%%%%%%%%%%%%%%%%%%%%%%%%%%%%%%%%%%%%%%%%%%
%%%%%%%%%%%%%%%%%%%%%%%%%%%%%%%%%%%%%%%%%%%%%%%%%%%%%%%%%%%%%%%%%%%%%%%%%%%%%%%%%%%%%%%%%
Quantum Steering is a form of correlation related to the ability of one part of a system 
to modify the state of a companion system when only local-measurements are performed on 
the former. 
More precisely, if through local measurements and classical communication 
one part of the system is able to convince the other part that they share an 
entangled state, the state is said to be steerable with respect to the 
first part \cite{wiseman}.

As in the previous subsection, 
considering the partition of the \hy{DF}{DF} as $ n = n_1 + n_2$, 
a state is {\it non}-Gaussian-steerable with respect to the first part 
(with $n_1$ \hy{DF}{DF}) if and only if \cite{wiseman2}
\begin{equation}                                                                         \label{gst}
{\bf V} + i \, {\bf \Pi}_{2} \ge 0, 
\,\,\, 
{\bf \Pi}_{2} := \tfrac{1}{2} (\mathsf J + {\bf T}_{\!n_2} \mathsf J {\bf T}_{\!n_2}), 
\end{equation}
with ${\bf T}_{\!n_2}$ defined in (\ref{pt}). 
In other words, it is not possible to steer the state of part 1, 
making local Gaussian measurements on part 2, if the last condition holds. 
The steering relation with respect to the second part 
is obtained by changing the roles of $n_1$ and $n_2$.  

As before, this concept can be related to the dynamical matrices,
and one can derive similar formulas by just replacing 
$ \mathsf J^{\!\top_{\!\! k}} \to {\bf \Pi}_{2} $ in (\ref{seprel}).
All Gaussian steerable states are entangled \cite{wiseman},
thus the set of matrices such that 
$\widetilde{\bf D}_{[i{\bf \Pi}_{2}]} \ge 0$ is a subset of the 
ones satisfying 
$\widetilde{\bf D}_{[i \mathsf J^{\!\top_{\!\! k}}]} \ge 0$.

\

This concludes our analysis of the bona-fide relations used across this paper.

%%%%%%%%%%%%%%%%%%%%%%%%%%%%%%%%%%%%%%%%%%%%%%%%%%%%%%%%%%%%%%%%%%%%%%%%%%%%%%%%%%%%%%%%%
%%%%%%%%%%%%%%%%%%%%%%%%%%%%%%%%%%%%%%%%%%%%%%%%%%%%%%%%%%%%%%%%%%%%%%%%%%%%%%%%%%%%%%%%%
\section{Symmetries of Steady-States} \label{ess}        %%%%%%%%%%%%%%%%%%%%%%%%%%%%%%%%
%%%%%%%%%%%%%%%%%%%%%%%%%%%%%%%%%%%%%%%%%%%%%%%%%%%%%%%%%%%%%%%%%%%%%%%%%%%%%%%%%%%%%%%%%
%%%%%%%%%%%%%%%%%%%%%%%%%%%%%%%%%%%%%%%%%%%%%%%%%%%%%%%%%%%%%%%%%%%%%%%%%%%%%%%%%%%%%%%%%
In principle, symmetries of the steady-states can be associated with the symmetries 
of the dynamics governed by the \hy{LME}{LME} \cite{albert,albert2}. 
In the perspective developed in this work, 
the relation of the steady-state symmetries 
and the symmetries of the dynamical matrices ${\bf \Gamma}$ and ${\bf D}$ 
will be investigated. 

Two \hy{LE}{LEs}, 
$\lfloor {\bf V},{\bf \Gamma}, {\bf D}   \rceil$ and
$\lfloor {\bf V}',{\bf \Gamma}', {\bf D}'   \rceil$, 
are said to be {\it covariant} when their matrices are related by  
\begin{equation}                                                                         \label{lyapcov}
{\bf V'} = \mathbf W {\bf V} \mathbf W^{\! \top } , \,\,\, 
{\bf \Gamma'} = \mathbf W {\bf  \Gamma} \mathbf W^{-1}, \,\,\, 
{\bf D'} = \mathbf W {\bf D} \mathbf W^{\! \top }, 
\end{equation}
for $ \mathbf W \in {\rm GL}(2n, \mathbb R)$. 
Note that, for an orthogonal $\bf W$, all above matrices 
are subjected to the same transformation. 
This covariance is helpful to determine the invariance properties of 
a steady-state, working directly at the level of the \hy{CM}{CM}:  

%%%%%%%%%%%%%%%%%%%%%%%%%%%%%%%%%%%%%%%%%%%%%%%%%%%%%%%
%%%%%%%%%%%%%%%%%%%%%%%%%%%%%%%%%%%%%%%%%%%%%%%%%%%%%%%
\begin{prop}                                                                              \label{prop:inv}
If $\mathbf \Gamma$ and ${\bf D}$ are invariant under 
${\bf W} \in {\rm GL}(2n,\mathbb R)$ [{\it i.e.}, 
${\bf W}  {\bf \Gamma} {\bf W}^{-1} = {\bf \Gamma}$
and
${\bf W} {\bf D} {\bf W}^{\!\top} = {\bf D}$], 
then ${\bf V}$ is invariant as well [{\it i.e.},  
${\bf W} {\bf V} {\bf W}^{\!\top} = {\bf V}$]. 
In this case, we say that $\mathbf W$ is a symmetry 
transformation of the Lyapunov equation. 
\end{prop}
%%%%%%%%%%%%%%%%%%%%%%%%%%%%%%%%%%%%%%%%%%%%%%%%%%%%%%%
\begin{proof}[Proof] Writing the solution 
${\bf V}({\bf \Gamma},{\bf D})$ as in Eq.(\ref{sol}) 
for the \hy{LE}{LE} 
$\lfloor {\bf V},{\bf \Gamma}, {\bf D} \rceil$, 
and since $\det {\bf W} \ne 0$, it is easy to see that 
${\bf V}({\bf \Gamma},{\bf D}) = 
 {\bf V}({\bf W}  {\bf \Gamma} {\bf W}^{-1},
 {\bf W} {\bf D} {\bf W}^{\top})$, 
{\it i.e.}, $\bf V$ remains unchanged.
\end{proof}
%%%%%%%%%%%%%%%%%%%%%%%%%%%%%%%%%%%%%%%%%%%%%%%%%%%%%%%
%%%%%%%%%%%%%%%%%%%%%%%%%%%%%%%%%%%%%%%%%%%%%%%%%%%%%%%

Let us give some particular but useful examples. Consider the transformation 
\begin{equation}                                                                         \label{transf1}
\mathbf W = \hat \sigma_z \otimes \mathsf I_2 = 
\left( 
\begin{array}{cc}
\mathsf{I}_2 & 0_2 \\ 
0_2 & - \mathsf{I}_2
\end{array} \right) %=  \mathbf W^{-1} 
\in {\rm GL}(4, \mathbb R). 
\end{equation}
If $\bf \Gamma$ and $\bf D$ are invariant under this transformation, 
they are necessarily written as 
${\bf \Gamma} = {\bf \Gamma}_1 \oplus {\bf \Gamma}_2$ and
${\bf D} = {\bf D}_1 \oplus {\bf D}_2$, where 
${\bf \Gamma}_j$, ${\bf D}_j \in {\rm Mat}(2,\mathbb R)$.
As a consequence of Prop.\ref{prop:inv}, ${\bf V} = {\bf V}_1 \oplus {\bf V}_2$, 
{\it i.e.}, it will be the \hy{CM}{CM} 
of a state without position-momentum correlations.
This invariance can be retrieved directly from the solution (\ref{sol}):
if $\bf A$ and $\bf Q$ are block-diagonals then $\bf P$ will also be 
block diagonal.  

Another example is the transformation 
\begin{equation}                                                                         \label{change}         
\mathbf W =  
\left( 
\begin{array}{cc}
0_2 & \mathsf{I}_2 \\ 
\mathsf{I}_2 & 0_2
\end{array} \right) %= \mathbf W^{-1} 
\in {\rm GL}(4, \mathbb R).                                   
\end{equation}
If $\bf \Gamma$ and $\bf D$ are invariant under this, 
the \hy{CM}{CM} will have momentum correlations equal to position correlations:
\begin{equation}                                                                         \label{invchange}
\mathbf V = 
\left(
     \begin{array}{cc}
     {\bf V}_{ \! 1} & {\bf V}_{12} \\ 
     {\bf V}_{12} & {\bf V}_{ \! 1}    
     \end{array}                     \right). 
\end{equation}

Focusing on the symplectic group, {\it i.e.}, 
choosing ${\bf W} \in \mathrm{Sp}(2n, \mathbb R) \subset {\rm GL}(2n, \mathbb R)$, 
we use the definition of $\bf \Gamma$ in (\ref{gammamat}) and the 
covariance relation (\ref{lyapcov}) to write 
\begin{equation}                                                                         \label{simcov}
{\bf V}' = {\bf W} {\bf V}{\bf W}^{\!\top}, \,  
{\bf H}' = {\bf W}^{-\top} \! {\bf H} {\bf W}^{-1}, \, 
\lambda' = {\bf W} {\lambda} . 
\end{equation}
This symplectic covariance, by the Stone-von Neumann theorem \cite{gosson}, 
is nothing but the representation of a unitary transformation of the \hy{LME}{LME}:
\begin{equation}                                                                         \label{unicov}
\hat \rho \longrightarrow \hat {\mathrm U} {\hat \rho} \hat{ \mathrm U}^{\dagger}, \,\,\,  
{\hat H}  \longrightarrow \hat {\mathrm U } {\hat H} \hat {\mathrm U}^{\dagger}, \,\,\, 
{\hat L_i} \longrightarrow \hat{\mathrm U}  {\hat L}_i \hat {\mathrm U}^{\dagger}, 
\end{equation}
where the unitary operator $\hat {\mathrm U}$ is the Metaplectic operator 
associated with ${\bf W} \in \mathrm{Sp}(2n, \mathbb R)$ \cite{nicacio,gosson}.

If we rearrange the elements of (\ref{cmdef}) consistently with the reordering 
$\hat x \mapsto \check x := (\hat q_1,\hat p_1,...,\hat q_n,\hat p_n )$ 
of (\ref{vecx}), then the invariance under the 
(nonreordered) transformation in (\ref{transf1}) implies that the steady-state of the 
system, with \hy{CM}{CM} ${\bf V} = {\bf V}_1 \oplus {\bf V}_2$,  
is the product state $\hat \rho = \hat \rho_1 \otimes \hat \rho_2$. 
In the reordered basis, the transformation (\ref{transf1}) is symplectic. 
Similarly, the (nonreordered) matrix in (\ref{change}) is also symplectic 
in the reordered basis, and it realizes the exchange of the subsystems. 
Consequently, the matrix in (\ref{invchange}) is the \hy{CM}{CM} of states with 
same local purity (symmetric states).  

As a last example, consider a symplectic rotation 
$\mathsf{R} \in {\rm K}(n) := {\rm Sp}(2n,\mathbb R) \cap {\rm O}(2n)$. 
As a consequence of its symplecticity and orthogonality, it is written as \cite{simon} 
\begin{equation}
\mathsf R =  
\left(
       \begin{array}{cc} 
          \mathbf Y  & \mathbf Z  \\
         -\mathbf Z & \mathbf Y
       \end{array}   
\right)                                                                                  \label{MO3}       
\end{equation} 
with $\mathbf Y, \mathbf Z \in {\rm Mat}(n,\mathbb R)$ satisfying the following 
conditions:
\begin{equation}
\mathbf Y\mathbf Y^{\!\top} + 
\mathbf Z\mathbf Z^{\!\top} = \mathsf{I}_n,  \,\,\,\,
\mathbf Y\mathbf Z^{\!\top} - 
\mathbf Z\mathbf Y^{\!\top} = 0_n .                                                        \label{MO5}       
\end{equation}
Any matrix written as
$ %\begin{equation}
{\bf M}: =  m_1 \mathsf I_{2n} + m_2 \mathsf J \in {\rm Mat}(2n,\mathbb R),  
$ %\end{equation}
with $m_1,m_2 \in \mathbb R$ is invariant under the whole group ${\rm K}(n)$.  
Note that, if ${\bf M} = {\bf M}^{\!\top}$, then $m_2 = 0$. 
If we consider 
\begin{equation}
{\bf \Gamma} = -\gamma_1 \mathsf I_{2n} + \gamma_2 \mathsf J, \,\,\,  
{\bf D} = \delta \mathsf I_{2n}, 
\end{equation}
{\it i.e.}, both invariant under ${\rm K}(n)$, then Prop.\ref{prop:inv} implies that
${\bf V} = \nu \, \mathsf I_{2n}$. 
By the other side, Corol.\ref{corol:diag} is a necessary and sufficient condition 
for this \hy{CM}{CM}, thus $\nu = \tfrac{\gamma}{2 \delta}$. 
It is important to mention that the matrices $\bf \Gamma$ and $\bf D$ on 
that corollary need not to be invariant. 

The subgroup of local rotations in $ {\rm K}(n)$ is described 
as the set of matrices (\ref{MO3}) with 
\begin{equation}
\mathbf Y = {\rm Diag}(y_{1}, y_{2}, ..., y_{n}),    \,\,\, 
\mathbf Z = {\rm Diag}(z_{1}, z_{2}, ..., z_{n}),
\end{equation}
which corresponds to a rotation 
\begin{equation}                                                                         \label{locrot}
{\sf R}_i := \begin{pmatrix}
              y_{i} & z_{i} \\ - z_{i} & y_{i}  
              \end{pmatrix}                    \in {\rm K}(1) 
\end{equation}
in each respective canonical pair $(\hat q_i, \hat p_i)$.  
The matrix ${\bf \Gamma} \in {\rm Mat}(2n,\mathbb R)$ is invariant under 
the local rotation subgroup if it is of the following form:
\begin{equation}                                                                         \label{locgamma}
{\bf \Gamma}  = \begin{pmatrix}
                {\bf \Gamma}_1 & {\bf \Gamma}_2 \\ - {\bf \Gamma}_2 & {\bf \Gamma}_1   
                \end{pmatrix}  
\end{equation}
with 
${\bf \Gamma}_i := {\rm Diag}(\gamma_{i1},...,\gamma_{in}) \in {\rm Mat}(n,\mathbb R)$. 
Since ${\bf D}$ is symmetric, it will be invariant under the same subgroup if it is 
written as 
\begin{equation}                                                                         \label{locd}
{\bf D} = {\rm Diag}(d_1,...,d_n,d_1,...,d_n) \in {\rm Mat}(2n,\mathbb R). 
\end{equation}
 
Assuming that $\bf \Gamma$ and $\bf D$ have this invariant structure, 
Prop.\ref{prop:inv} guaranties that the \hy{CM}{CM} of the steady-state will be 
\begin{equation}                                                                         \label{loccm}               
{\bf V} = {\rm Diag}(v_1,...,v_n,v_1,...,v_n), \,\, v_i \in \mathbb R \, \forall i, 
\end{equation}
which is the \hy{CM}{CM} of $n$-mode thermal state.  

%%%%%%%%%%%%%%%%%%%%%%%%%%%%%%%%%%%%%%%%%%%%%%%%%%%%%%%%%%%%%%%%%%%%%%%%%%%%%%%%%%%%%%%%%
%%%%%%%%%%%%%%%%%%%%%%%%%%%%%%%%%%%%%%%%%%%%%%%%%%%%%%%%%%%%%%%%%%%%%%%%%%%%%%%%%%%%%%%%%
\section{ Examples }\label{eI}                       %%%%%%%%%%%%%%%%%%%%%%%%%%%%%%%%%%%%
%%%%%%%%%%%%%%%%%%%%%%%%%%%%%%%%%%%%%%%%%%%%%%%%%%%%%%%%%%%%%%%%%%%%%%%%%%%%%%%%%%%%%%%%%
%%%%%%%%%%%%%%%%%%%%%%%%%%%%%%%%%%%%%%%%%%%%%%%%%%%%%%%%%%%%%%%%%%%%%%%%%%%%%%%%%%%%%%%%%
Let us now present some examples to show the usefulness of the results presented 
in this work.    
\vspace{-1.5cm}
%%%%%%%%%%%%%%%%%%%%%%%%%%%%%%%%%%%%%%%%%%%%%%%%%%%%%%%%%%%%%%%%%%%%%%%%%%%%%%%%%%%%%%%%%
\subsection{Two Oscillators interacting with Thermal Baths }\label{to} %%%%%%%%%%%%%%%%%%
%%%%%%%%%%%%%%%%%%%%%%%%%%%%%%%%%%%%%%%%%%%%%%%%%%%%%%%%%%%%%%%%%%%%%%%%%%%%%%%%%%%%%%%%%
\vspace{-0.29cm}
Consider two coupled harmonic oscillators, each one interacting with its own 
thermal bath. The frequency of the oscillators are 
$\omega_1$ and $\omega_2$ and the spring constant is $\kappa$.  

The Hamiltonian of the system is given by (\ref{ham-lind}) with $\xi = 0$, $H_0 = 0$ 
and 
\begin{equation}                                                                         \label{hess1}
{\bf H} = \left[ 
                 \begin{array}{cc}
                 \omega_1 + \frac{\kappa}{2} & -\tfrac{\kappa}{2} \\
                 -\tfrac{\kappa}{2} & \omega_2 + \frac{\kappa}{2}
                 \end{array}
         \right]  \oplus 
          \left[ 
                 \begin{array}{cc}
                 \omega_1 & 0 \\
                 0        & \omega_2
                 \end{array}
         \right]. 
\end{equation}
The coupling between a given oscillator and the respective
reservoir is described by the Lindblad operators \cite{wiseman}
\begin{equation}                                                                         \label{tb}
\!\! \hat L_k  = \sqrt{\hbar\zeta_k (\bar N_k + 1)} \, \hat a_k , \, 
\hat L_k' = \sqrt{\hbar\zeta_k \bar N_k} \, \hat a_k^\dagger, \, k = 1,2, 
\end{equation}
where $\zeta_k \ge 0$ are the bath-oscillator couplings, 
$\bar N_k \ge 0$ are thermal occupation numbers, 
and $\hat a_k\! := \!(\hat q_k + i \hat p_k)/\sqrt{2\hbar}$ 
is the annihilation operator associated with mode $k$.  
This choice for the reservoirs and Eq.(\ref{ham-lind}) allow us to identify  %
\begin{equation}                                                                         \label{lintb}                                                                       
\begin{aligned} 
{\lambda}_1 & =  \sqrt{\tfrac{\zeta_1}{2} (\bar N_1 + 1)} (i, 0, -1, 0 )^\top,  \\
\lambda'_1  & =  \sqrt{\tfrac{\zeta_1}{2} \bar N_1} ( i , 0 , -1, 0 )^\dagger,  \\
{\lambda}_2 & =  \sqrt{\tfrac{\zeta_2}{2} (\bar N_2 + 1)} (0, i, 0, -1 )^\top,  \\
\lambda'_2  & =  \sqrt{\tfrac{\zeta_2}{2} \bar N_2} ( 0 , i , 0 , -1 )^\dagger. 
\end{aligned} \vspace{0cm}
\end{equation} 
With the above vectors, 
and using Eqs.(\ref{gammamat}), (\ref{upsilon}) and (\ref{dmat}), one finds
\begin{equation}                                                                         \label{dynsys}
\begin{aligned}
&\!\!{\bf D} = \pmb{D} \oplus \pmb{D}, \, 
\pmb{D} := {\rm Diag}[ \zeta_1 (2\bar{N}_{1} + 1),\zeta_2 (2\bar{N}_{2} + 1)], \\
&\!\!{ \bf \Gamma }  = 
\left[ \begin{array}{cccc}
        -\frac{\zeta_1}{2} & 0 & \omega_1 & 0 \\
         0                  &  -\frac{\zeta_2}{2} & 0  & \omega_2  \\
        -\omega_1 - \frac{\kappa}{2}  & \frac{\kappa}{2}  &  - \frac{\zeta_1}{2} & 0 \\
         \frac{\kappa}{2}  & - \omega_2 - \frac{\kappa}{2}&  0 &- \frac{\zeta_2}{2}
       \end{array} 
\right].
\end{aligned}\vspace{0cm}
\end{equation}
For simplicity, we will consider 
$ \zeta_1 = \zeta_2 = \zeta$, $\omega_1 = \omega_2 = \omega$,  
and the eigenvalues of $\bf \Gamma$ become 
\begin{equation}
{\rm Spec}_{\mathbb C}({\bf \Gamma}) =  
  \left\{ -\tfrac{\zeta }{2} \pm i \omega, 
 -\tfrac{\zeta }{2} \pm i \sqrt{\omega(\omega + \kappa)  } \right\}.  
\end{equation}
As one can see, $\bf \Gamma$ is \hy{AS}{AS} since 
${\rm Re}\left[ {\rm Spec}_{\mathbb C}({\bf \Gamma}) \right] < 0$ 
for all (positive) values of the parameters.   

The separability of the steady-state will be retrieved from Eqs.(\ref{seprel}). 
Since $\bf \Gamma \ne {\bf \Gamma}^{\!\top}$, 
Eq.(\ref{seprel:1}) will be applied and, calculating the eigenvalues of 
${\bf D}_{ [ i \mathsf J^{\!\top_{\!\! k} } ] }$, one finds
\begin{widetext}
\begin{equation}
{\rm Spec}_{\mathbb R} ({\bf D}_{ [ i \mathsf J^{\!\top_{\!\! 2} } ] } ) = 
\left\{ 
\zeta (\bar{N}_1 + \bar{N}_2 + 1) \pm 
\sqrt{ \tfrac{\kappa^2}{2} + \zeta^2 (\bar{N}_1 - \bar{N}_2)^2 + \zeta^2  \pm 
\sqrt{ \tfrac{\kappa^4}{4} + \zeta^2 \kappa^2 + 4\zeta^4 (\bar{N}_1 - \bar{N}_2)^2 }}
\right\}.   
\end{equation}
\end{widetext}
The steady-state is separable if this spectrum is non-negative,  
or explicitly when  
\begin{equation}                                                                         \label{sepex1}                                                                         
\frac{\zeta}{\kappa} \ge   
   \mathcal S( \bar{N}_1 , \bar{N}_2 ) := \sqrt{ 
   \frac{ ( 2\bar{N}_1 + 1 )( 2 \bar{N}_2 + 1 ) }
        { 16 \bar{N}_1 \bar{N}_2 (\bar{N}_1 + 1) (\bar{N}_2 + 1)}  }.                           
\end{equation}
In this example, states with 
$ 0 \in {\rm Spec}_{\mathbb R}({\bf D}_{ [ i \mathsf J^{\!\top_{\!\! k} } ] } )$ 
are all lying in the surface 
$\zeta/\kappa = \mathcal S( \bar{N}_1 , \bar{N}_2 ) $.  
The classicality of the steady-state will be determined by Eq.(\ref{clas}) with
\begin{equation}
\begin{aligned}
&{\rm Spec}_{\mathbb R} ({\bf D}_{ [ -\mathsf I ] } ) = \\
&\left\{ 
\zeta (\bar{N}_1 + \bar{N}_2 ) \pm  \tfrac{\kappa}{2} \pm
 \sqrt{ \tfrac{\kappa^2}{4} + \zeta^2 (\bar{N}_1 - \bar{N}_2)^2 }
\right\},
\end{aligned}
\end{equation}
and the steady-state will be classical if 
\begin{equation}                                                                          \label{clasex1}
\frac{\zeta}{\kappa} \ge \mathcal P (\bar{N}_1 , \bar{N}_2):=
\frac{\bar{N}_1 + \bar{N}_2  }{4 \bar{N}_1 \bar{N}_2   }.  
\end{equation}
In Fig.\ref{fig1} 
we show the functions in (\ref{sepex1}) and (\ref{clasex1}).  
Since a classical state is always separable, 
$\mathcal S(\bar{N}_1 , \bar{N}_2) < \mathcal P(\bar{N}_1 , \bar{N}_2)$. 

%%%%%%%%%%%%%%%%%%%%%%%%%%%%%%%%%%%%%%%%%%%%%%%%%%%%%%%%%%%%%%%%%%%%%
%%%%%%%%%%%%%%%%%%%%%%%%%%%%%%%%%%%%%%%%%%%%%%%%%%%%%%%%%%%%%%%%%%%%%
\begin{figure}[!htb] 
\centering
\includegraphics[width=7.2cm,trim=0 29 0 30]{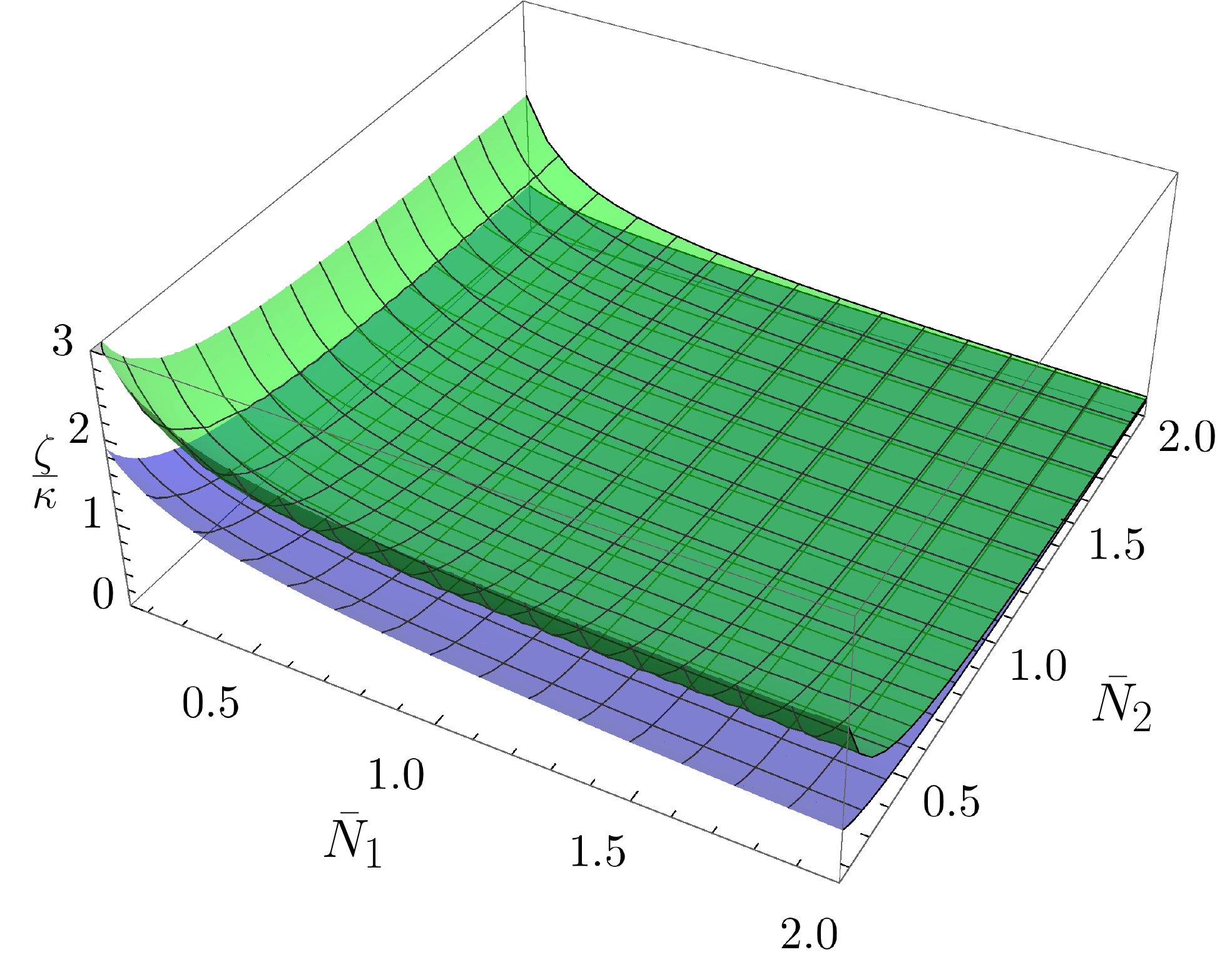} 
\caption{%
(Color Online) Separability surface $\mathcal S$ (bottom) and 
Classicality surface $\mathcal P$ (top).
States either on or above $\mathcal S$ are separable while 
states either on or above $\mathcal P$ are classical.   
Between $\mathcal S$ and $\mathcal P$ there are classical
and nonclassical separable states. 
Bellow $\mathcal S$, there are separable and entangled states. 
The divergence of both functions for $\bar N_1 \to 0$ and for 
$\bar N_2 \to 0$ are not shown in this plot. 
}                                                                                        \label{fig1} 
\end{figure}
%%%%%%%%%%%%%%%%%%%%%%%%%%%%%%%%%%%%%%%%%%%%%%%%%%%%%%%%%%%%%%%%%%%%% 
%%%%%%%%%%%%%%%%%%%%%%%%%%%%%%%%%%%%%%%%%%%%%%%%%%%%%%%%%%%%%%%%%%%%% 

To understand the sufficiency of the results for the system in consideration,
as a consequence of the fact that ${\bf \Gamma} \ne {\bf \Gamma}^{\!\top}$, 
we will explore the separability and classicality criteria directly applying both to the 
\hy{CM}{CM} of the steady-state. 
Using Eqs.~(\ref{dynsys}), we are able to obtain analytically the
solution ${\bf V} = {\bf P}({\bf \Gamma},{\bf D})$ using (\ref{sol}) 
or solving algebraically the \hy{LE}{LE} $\lfloor {\bf V}, {\bf \Gamma}, {\bf D} \rceil$.  
For simplicity and without loss of generality, we choose $\bar N_1 = \bar N_2 = \bar N$ 
and the solution is 
\begin{equation}                                                                         \label{sscm}
\begin{aligned}
{\bf V} &= (2\bar N +1) \mathsf{I}_{4} \, +   \\ 
&  \frac{(2\bar N +1)\kappa}{\zeta^2 + 4\omega(\omega+\kappa)}
\left[ \begin{array}{cc}
  \omega & \tfrac{1}{2} \zeta  \\
  \tfrac{1}{2} \zeta & (\omega+\kappa) 
 \end{array} \right]\! \otimes \!
 \left[\! \begin{array}{rr}
  -1 & 1  \\
   1 & -1 
 \end{array} \!\right].
\end{aligned}
\end{equation}
The steady-state is classical if and only if the above \hy{CM}{CM} obeys (\ref{bfclas}), 
or working out its eigenvalues, if and only if
\begin{equation}                                                                         \label{clasex1a}                                                                         
\frac{\zeta}{\kappa} \ge   
   \mathcal P'( \bar{N} ) := \sqrt{ \frac{1}{4\bar{N}^{2}} - 
                                    \left(\frac{2\omega}{\kappa} + 1\right)^2   }.                           
\end{equation}
As for separability, the steady state will be separable if and only if 
the condition (\ref{bfsep}) is satisfied, which reads as 
\begin{equation}                                                                         \label{sepex1a}                                                                         
\frac{\zeta}{\kappa} \ge   
   \mathcal S'( \bar{N} ) := \sqrt{ 
                                    \frac{ 1 } { 16 \bar{N}^2(\bar{N} + 1)^2} -
                                     \left(\frac{2\omega}{\kappa} + 1\right)^2  }.                           
\end{equation}
In Fig.~\ref{fig2}, the two functions representing the necessary and sufficient 
conditions in Eqs.~(\ref{clasex1a}) and~(\ref{sepex1a}) are shown. 
We also compare them with the two sufficient conditions 
(\ref{sepex1}) and (\ref{clasex1}), already plotted in Fig.~\ref{fig1}. 
It is clear that, for the system in question, even if the sufficient criteria become 
tighter for smaller values of $\bar{N}$, they are in general unable to determine 
whether a state is entangled or nonclassical. 

%%%%%%%%%%%%%%%%%%%%%%%%%%%%%%%%%%%%%%%%%%%%%%%%%%%%%%%%%%%%%%%%%%%%%
%%%%%%%%%%%%%%%%%%%%%%%%%%%%%%%%%%%%%%%%%%%%%%%%%%%%%%%%%%%%%%%%%%%%%
\begin{figure}[!htbp] 
\centering
\includegraphics[width=8.cm,trim=0 0 0 0]{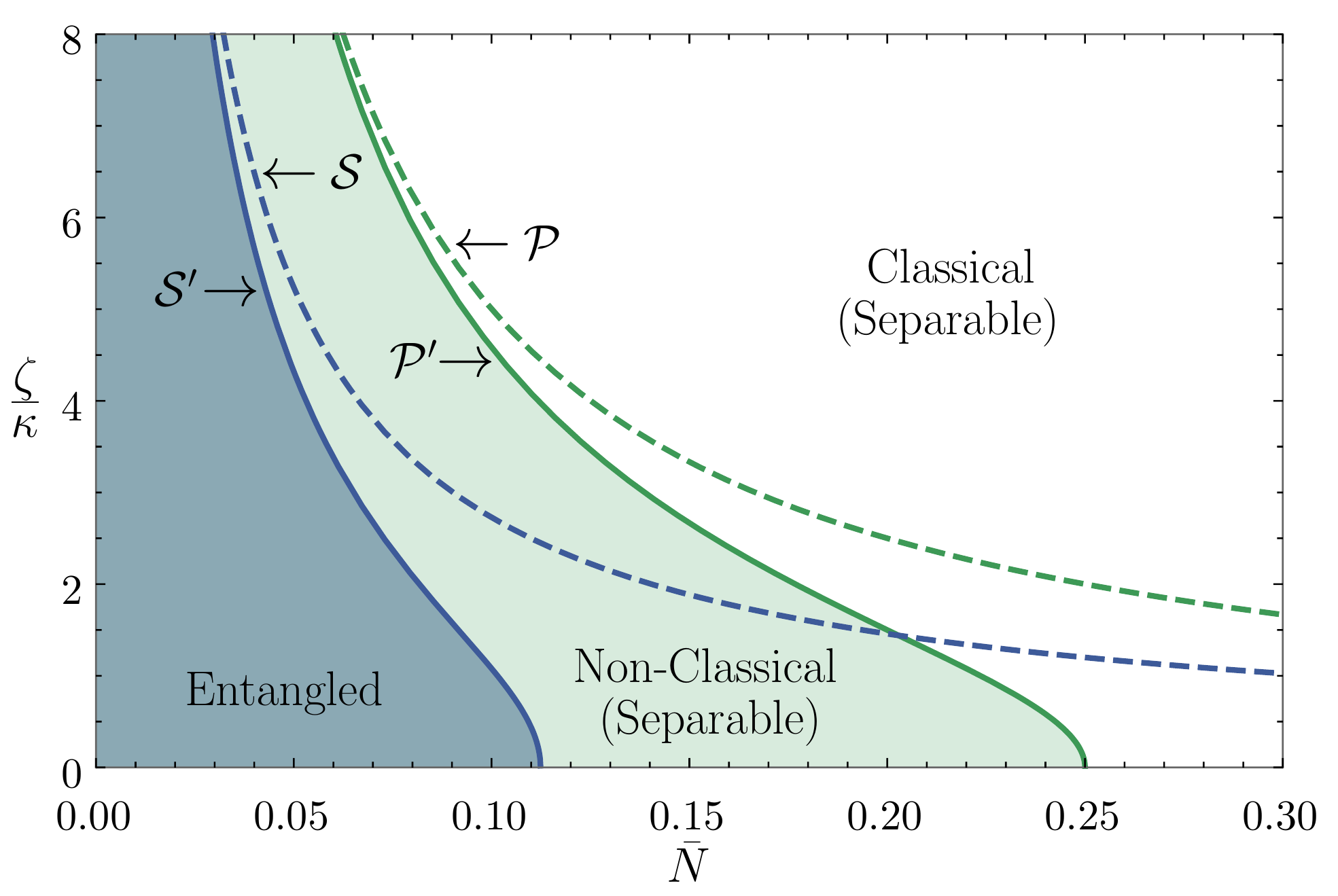} 
\caption{%
(Color Online) Necessary and sufficient conditions for classicality 
and for separability of the steady-state. 
We show the functions $\mathcal P'( \bar{N} )$ in (\ref{clasex1a}) 
and $\mathcal S'( \bar{N} )$ in (\ref{sepex1a}), 
as indicated in the graph. 
The regions limited by these two functions indicates the 
nature of the steady-state. 
We also plot the functions 
$\mathcal S(\bar{N},\bar{N})$ in (\ref{sepex1}) and 
$\mathcal P(\bar{N},\bar{N})$ in (\ref{clasex1}) 
as dashed curves, see Fig.~\ref{fig1}. In this plot $\omega/\kappa = 0.5$. 
}                                                                                        \label{fig2} 
\end{figure}
%%%%%%%%%%%%%%%%%%%%%%%%%%%%%%%%%%%%%%%%%%%%%%%%%%%%%%%%%%%%%%%%%%%%% 
%%%%%%%%%%%%%%%%%%%%%%%%%%%%%%%%%%%%%%%%%%%%%%%%%%%%%%%%%%%%%%%%%%%%% 

For the Gaussian-steering property, we also consider the case $\bar N_1=\bar N_2 =\bar N$
and the sufficient condition is determined by calculating the matrix 
${\bf D}_{ [ i{\bf \Pi}_{k} ]}$ 
[the matrix ${\bf \Pi}_{k}$ is defined in (\ref{gst}), and here ${n_1} = {n_2} = 1$], 
which has the eigenvalues  
\begin{eqnarray}
&&\!\!\!\!\!\!{\rm Spec}_{\mathbb R} (
{\bf D}_{ [ i{\bf \Pi}_{1}] } )  =  
{\rm Spec}_{\mathbb R} (
{\bf D}_{ [ i{\bf \Pi}_{2}] } )  =       \nonumber \\
&&\!\!\!\!\!\! \left\{ ( 2 \bar{N} +1), \zeta( 2 \bar{N} +1) \pm \sqrt{4\zeta^2 + \kappa^2}   
\right\}.    
\end{eqnarray}
The steady-state will be {\it non}-Gaussian steerable with respect to both partitions if 
$\zeta/\kappa \ge [4 (2\bar{N} + 1)^2 + 4]^{-1/2}$. 

%%%%%%%%%%%%%%%%%%%%%%%%%%%%%%%%%%%%%%%%%%%%%%%%%%%%%%%%%%%%%%%%%%%%%%%%%%%%%%%%%%%%%%%%%
\subsection{Two oscillators interacting with thermal baths in RWA}\label{tor}  %%%%%%%%%%
%%%%%%%%%%%%%%%%%%%%%%%%%%%%%%%%%%%%%%%%%%%%%%%%%%%%%%%%%%%%%%%%%%%%%%%%%%%%%%%%%%%%%%%%%
Let us consider the same system as before, but with the Hamiltonian 
\begin{equation}                                                                         \label{hess2}
{\bf H} = \left[ 
                 \begin{array}{cc}
                 \varpi_1  & \Omega \\
                 \Omega    & \varpi_2
                 \end{array}
         \right]  \oplus 
          \left[ 
                  \begin{array}{cc}
                 \varpi_1  & \Omega \\
                 \Omega  &   \varpi_2 
                 \end{array}
         \right]. 
\end{equation}
This Hamiltonian is derived from the one in Eq.(\ref{hess1}) by applying 
a rotating wave approximation (RWA).
The procedure and the validity of this result are carefully 
discussed in the Appendix of Ref.\cite{nicacio2}, 
as well as the relation among the coupling constant $\Omega$ in Eq.(\ref{hess2}) 
with the parameters in Eq.(\ref{hess1}), 
see \cite{footnote2} for details. 
Considering also the same structure for the reservoirs in (\ref{lintb}), one finds
\begin{equation}                                                                         \label{dynsys2}
{ \bf \Gamma }  = 
          \left[ 
                 \begin{array}{cccc}
                 -\tfrac{\zeta_1}{2}  & 0 & \varpi_1 & \Omega  \\
                 0 & -\tfrac{\zeta_2}{2} &  \Omega & \varpi_2 \\
                 -\varpi_1  & -\Omega   &     -\tfrac{\zeta_1}{2} & 0 \\
                 -\Omega    & -\varpi_2 & 0 & -\tfrac{\zeta_2}{2}
                 \end{array}
         \right],   
\end{equation}
which is \hy{AS}{AS} with eigenvalues  
\begin{equation}
\!\!\! {\rm Spec}_{\mathbb C}({\bf \Gamma}) \!= \! 
\left\{ 
         -\tfrac{\zeta_1+\zeta_2}{4} \pm 
         \tfrac{1}{4}\sqrt{(\zeta_1 - \zeta_2)^2 - 4 \Omega^2} \pm i \varpi
                                                                            \right\},  
\end{equation}
where we used $\varpi_1 = \varpi_2 =: \varpi$. 
Note that $\bf \Gamma$ in (\ref{dynsys2}) and $\mathbf D$ in (\ref{dynsys}) 
are invariant under a rotation by $\mathsf J \in {\rm K}(2)$. 
Following Prop.\ref{prop:inv}, 
the steady-state \hy{CM}{CM} for this system will also be, {\it i.e.}, 
$ \mathbf V = \mathsf J \mathbf V \!\mathsf J^{\!\top}$, 
which means that position-momentum correlations are antisymmetric and 
position-position correlations are equal to momentum-momentum correlations. 
This symmetry help us to solve algebraically the \hy{LE}{LE} (\ref{lyapeq2}), 
obtaining \cite{assadian} 
\begin{equation}                                                                         \label{cmex2}
{\bf V} = \left[ 
                 \begin{array}{cccc}
                 v_{1}    & 0        & 0       & v_{14}     \\
                 0        & v_{2}    & -v_{14} & 0          \\
                 0        & - v_{14} &  v_{1}  & 0          \\
                 v_{14}   & 0        & 0       & v_{2}      \\
                 \end{array} \right],   
\end{equation}
with 
\begin{eqnarray}
v_1 &=& 2 \frac{\zeta_1\bar{N}_1 + \zeta_2\bar{N}_2}{\zeta_1 + \zeta_2} +  
          \frac{2(\bar{N}_1-\bar{N}_2)\zeta_1\zeta_2^2} 
               {(\zeta_1 + \zeta_2)(4\Omega^2 + \zeta_1\zeta_2)}+1,                      \nonumber\\
v_2 &=&  2 \frac{\zeta_1\bar{N}_1 + \zeta_2\bar{N}_2}{\zeta_1 + \zeta_2} + 
           \frac{2(\bar{N}_2-\bar{N}_1)\zeta_1^2\zeta_2} 
               {(\zeta_1 + \zeta_2)(4\Omega^2 + \zeta_1\zeta_2)} + 1,                    \nonumber\\
v_{14} &=& \frac{4\zeta_1\zeta_2\Omega (\bar{N}_2 - \bar{N}_1)}
                {(\zeta_1+\zeta_2)(\zeta_1\zeta_2 + 4\Omega^2)}. 
\end{eqnarray}
Note that, if $\zeta_2 = 0$, then ${\bf V} = (2\bar{N}_1 + 1) \mathsf I_4$, 
which has a simple structure, but it can not be simply retrieved by symmetries of 
$\bf \Gamma$ and $\bf D$. 

The simple form of (\ref{cmex2}) can be used to explicitly analyze the results in 
conditions in (\ref{clas}).
Considering for simplicity $ \zeta_1 = \zeta_2 = \zeta$, 
the (doubly degenerate) spectrum of (\ref{cmex2}) is 
\begin{equation}                                                                         \label{espex2}
{\rm Spec}_{\mathbb R} ({\bf V})  =  
 \left\{  \bar{N}_1 + \bar{N}_2 + 1 \pm 
          \frac{(\bar{N}_1 - \bar{N}_2)} {\sqrt{1+4\Omega^2/\zeta^2 }} \right\}.   
\end{equation}
From this, it is easy to see that the state (\ref{cmex2}) is classical 
for any values of the parameters, since ${\bf V} - \mathsf I_{4} \ge 0$. 
%, and the equality is attained only when $\bar{N}_1 = \bar{N}_2 = 0$. 
%
On the other hand, let us calculate
\begin{equation}                                                                         \label{clasex2}
{\rm Spec}_{\mathbb R}({\bf D}_{ [ -\mathsf I ] } )  =  
\left\{ 2\zeta \bar{N}_1, 2\zeta \bar{N}_2 \right\},
\end{equation}
which is non-negative for any value of $\bar{N}_1$ and $\bar{N}_2$. 
The statement in (\ref{clas}) thus tells us that all steady-states of this system 
belong to the set of classical states, 
which was already found by means of Eq.(\ref{espex2}). 
%
%%%%%%%%%%%%%%%%%%%%%%%%%%%%%%%%%%%%%%%%%%%%%%%%%%%%%%%%%%%%%%%%%%%%%%%%%%%%%%%%%%%%%%%%%
\subsection{Cascaded OPO}\label{opo}      %%%%%%%%%%%%%%%%%%%%%%%%%%%%%%%%%%%%%%%%%%%%%%%
%%%%%%%%%%%%%%%%%%%%%%%%%%%%%%%%%%%%%%%%%%%%%%%%%%%%%%%%%%%%%%%%%%%%%%%%%%%%%%%%%%%%%%%%%
Consider an optical parametric oscillator (OPO) coupled to the vacuum field \cite{koga}. 
The Hamiltonian is written as 
$\hat H = i\hbar\epsilon ({\hat a^{\dagger 2}} - {\hat a}^2)/4$, 
where $\epsilon \ge 0$ denotes the effective pump intensity. 
The coupling with the vacuum is described by the operator 
$\hat L = \sqrt{\hbar\kappa} \hat a $, where $\kappa >0$ is 
the damping cavity rate.

With the help of Eqs.~(\ref{ham-lind}), (\ref{gammamat}), (\ref{upsilon}), 
and~(\ref{dmat}), one readily reaches 
\begin{equation}                                                                         \label{opo1}
{ \bf \Gamma }  =
\left[ \begin{array}{cc}
        \frac{1}{2}(\epsilon -\kappa) & 0  \\
        0         & -\frac{1}{2}(\epsilon +\kappa)   
       \end{array}
\right], \,\,\,
{\bf D} = \kappa \, \mathsf I_2.  
\end{equation}
The matrix $\bf \Gamma = {\bf \Gamma}^{\!\top}$ 
will be \hy{AS}{AS} in so far as $\kappa > \epsilon$.  
Also, ${\bf D} >0$, once $\kappa >0$. Following the statement in (\ref{clas2}), 
since $\widetilde{\bf D}_{[i \mathsf J^{\!\top_{\!\! 2}}]} = 
           {\rm Diag}\left(\epsilon, -\epsilon \right)$, 
the steady-state will be nonclassical if $\epsilon \ne 0$ and 
will be classical only if $\epsilon = 0$. 
The \hy{LE}{LE} is trivially solved to give the \hy{CM}{CM} of the steady-state:
\begin{equation}
{\bf V} = {\rm Diag}[\kappa/( \kappa -\epsilon),\kappa/( \kappa +\epsilon)],  
\end{equation}
which is a squeezed thermal state and corresponds to the coherent-state
solution in \cite{koga} if $\epsilon = 0$. 

For a cascaded OPO \cite{koga}, 
the system is described by the Hamiltonian 
$\hat H = \hat H_1 + \hat H_2 + \frac{1}{2i}(\hat L_1^\dagger \hat L_2 
        - \hat L_2^\dagger \hat L_1) $ 
with 
$\hat H_j = i\hbar\epsilon_j ({\hat a_j^{\dagger 2}} - {\hat a_j}^2)/4$    
and $\hat L_j = \sqrt{\hbar\kappa} \hat a_j $.
The coupling of the system with the intracavity vacuum is 
represented by $\hat L = \hat L_1 + \hat L_2$. %
Under these circumstances we write
\begin{equation}                                                                         \label{opod}
{\bf D} =
          \left[ 
                 \begin{array}{cc}
                 \kappa  & \kappa \\
                 \kappa  & \kappa
                 \end{array}
         \right]  \oplus 
          \left[ 
                 \begin{array}{cc}
                 \kappa  & \kappa \\
                 \kappa  & \kappa
                 \end{array}
         \right], 
\end{equation}
and 
\begin{equation}                                                                         \label{opog}
{ \bf \Gamma }  =  
         \left[\!\! 
                 \begin{array}{cc}
                  \frac{\epsilon_1-\kappa}{2} & 0 \\
                  -\kappa      &  \frac{\epsilon_2-\kappa}{2}
                 \end{array}\!\!
         \right] \! \oplus \!
          \left[\!\! 
                  \begin{array}{cc}
                   -\tfrac{\epsilon_1+\kappa}{2}  & 0   \\
                  -\kappa    & -\tfrac{\epsilon_2+\kappa}{2}
                 \end{array}\!\!
         \right] 
\end{equation}
with spectrum given by 
\begin{equation}
{\rm Spec}({\bf \Gamma}) =  
\left\{ -\tfrac{1}{2} (\kappa \pm \epsilon_1),
        -\tfrac{1}{2} (\kappa \pm \epsilon_2) \right\}.  
\end{equation}
Note that ${\bf \Gamma} \ne {\bf \Gamma}^{\!\top}$ is \hy{AS}{AS} if 
$\kappa > \max\{|\epsilon_1|,|\epsilon_2|\}$.  
Now the eigenvalues of ${\bf D}_{[i \mathsf J^{\!\top_{\!\! 2}}]}$ 
in (\ref{seprel}) can be calculated, yielding
\begin{equation}
{\rm Spec}_{\mathbb R} ({\bf D}_{[i \mathsf J^{\!\top_{\!\! 2}}]}) = 
\{ 
( 1 \pm \sqrt{5})\kappa, 2\kappa, 0\},  
\end{equation}
which shows that the state will be always entangled 
following Corol.\ref{cor:if}. 
The same arguments can be applied to determine the 
Gaussian steerability. Calculating the matrix 
${\bf D}_{ [ i{\bf \Pi}_{k} ]}$ 
[the matrix ${\bf \Pi}_{k}$ is defined in (\ref{gst}), and here ${n_1} = {n_2} = 1$], 
which has the eigenvalues  
\begin{equation}
\begin{aligned}
{\rm Spec}_{\mathbb R} ({\bf D}_{ [ i{\bf \Pi}_{1} ] } ) & =  
\left\{   ( 1 \pm \sqrt{5} ) \frac{ \kappa }{ 2 }, 
          ( 3 \pm \sqrt{5} ) \frac{ \kappa }{ 2 }   \right\}, \\        
{\rm Spec}_{\mathbb R} ({\bf D}_{ [ i{\bf \Pi}_{2} ] } ) & =  
\left\{   ( 3 \pm \sqrt{17} ) \frac{ \kappa }{ 2 }, k , 0     \right\}.   
\end{aligned}
\end{equation}
From these, one concludes that the state will be always Gaussian steerable with respect 
to both modes since these spectra is nonpositive. 

Following Prop.\ref{prop:inv}, the \hy{CM}{CM} of the steady-state of the cascaded OPO 
will have the symmetry induced by (\ref{transf1}), and reads
\begin{equation}                                                                         \label{opocm}
{\bf V}  = \left[\begin{array}{cc}
                 \frac{k}{k-\epsilon_1}     & -\frac{2\kappa \epsilon_1}{g_{-}} \\
                 - \frac{2\kappa \epsilon_1}{ g_{-}} & -\frac{\kappa  h_+}{ g_-}
                 \end{array}\right] \oplus 
           \left[\begin{array}{cc}
                 \frac{k}{k+\epsilon_1}     & \frac{2\kappa \epsilon_1}{ g_{+}} \\
                 \frac{2\kappa \epsilon_1}{ g_{+}}   & \frac{\kappa h_-}{ g_+}
                 \end{array}\right],
\end{equation}
where we have defined 
\begin{eqnarray*}
g_{\pm} &=& (\epsilon_1 + \epsilon_2 \pm 2 \kappa)(\epsilon_1\pm \kappa) , \nonumber \\
h_{\pm} &=& (\epsilon_1^2 + \epsilon_1 \epsilon_2 \pm \epsilon_1 \kappa 
             + 2 \kappa^2 \mp \kappa \epsilon_2)(\epsilon_2\mp k)^{-1}.  
\end{eqnarray*}
%%%%%%%%%%%%%%%%%%%%%%%%%%%%%%%%%%%%%%%%%%%%%%%%%%%%%%%%%%%%%%%%%%%%%%%%%%%%%%%%%%%%%%%%%
\subsection{OPO and Thermal Baths}\label{otb}        %%%%%%%%%%%%%%%%%%%%%%%%%%%%%%%%%%%%
%%%%%%%%%%%%%%%%%%%%%%%%%%%%%%%%%%%%%%%%%%%%%%%%%%%%%%%%%%%%%%%%%%%%%%%%%%%%%%%%%%%%%%%%%
Consider the Hamiltonian dynamics of two particles described by
\begin{equation}
\hat H = \frac{\epsilon}{4} {\pmb\{}\hat q_1,\hat p_1{\pmb\}}_{\!+} + 
         \frac{\epsilon}{4} {\pmb\{}\hat q_2,\hat p_2{\pmb\}}_{\!+} +  
         \frac{\kappa}{2} (\hat q_2 \hat p_1 + \hat p_2 \hat q_1).   
\end{equation}
This Hamiltonian is similar to the one in the previous example, 
it is basically the Hamiltonian of the cascaded OPO 
with a phase change \cite{wiseman}. 
If the particles are in contact with the thermal baths, 
as in (\ref{tb}), the dynamical matrices become 
\begin{equation}
{ \bf \Gamma }  =  {{\bf \Gamma}^{\!\top}} = 
        \tfrac{1}{2} \!\left[\!\! 
                 \begin{array}{cc}
                  \epsilon - \zeta & k \\
                  k      &  \epsilon-\zeta
                 \end{array}\!\!
         \right] \! \oplus 
         \tfrac{1}{2}\! \left[\!\! 
                  \begin{array}{cc}
                 - (\epsilon + \zeta) & k \\
                  k      & - (\epsilon+\zeta)
                 \end{array}\!\!
         \right],  
\end{equation}
and $\bf D$ is as in (\ref{dynsys}), 
and now we are considering $\bar N_1 = \bar N_2 = \bar N$ and $\zeta_1 = \zeta_2 = \zeta$.
The eigenvalues of $\bf \Gamma$ are
\begin{equation}
{\rm Spec}({\bf \Gamma}) =  
\left\{ -\tfrac{1}{2} \zeta \pm \tfrac{1}{2}( \epsilon + \kappa),
        -\tfrac{1}{2} \zeta \pm \tfrac{1}{2}( \epsilon - \kappa) \right\},   
\end{equation}
and it will be \hy{AS}{AS} if $\zeta > \epsilon + \kappa$.  
The eigenvalues of ${\bf D}_{[-\mathsf I]}$ in (\ref{clas2}) are 
\begin{equation}
{\rm Spec}_{\mathbb R} ({\bf D}_{ [ -\mathsf I ] } ) = 
\left\{ 
2\zeta \bar{N} \pm (\epsilon + \kappa), 2\zeta \bar{N} \pm (\epsilon - \kappa) 
\right\},
\end{equation}
which shows that the state will be classical if and only if 
$2\zeta \bar{N} \ge (\epsilon + \kappa)$. 
The eigenvalues of $\widetilde{\bf D}_{[i \mathsf J^{\!\top_{\!\! 2}}]}$ 
in (\ref{seprel:2}) are 
\begin{equation}
{\rm Spec}_{\mathbb R} (\widetilde{\bf D}_{[i \mathsf J^{\!\top_{\!\! 2}}]}) = 
\{ 
2\zeta \bar{N} \pm \kappa, 2\zeta(\bar{N} + 1) \pm \kappa \}, 
\end{equation}
and the steady-state will be entangled if and only if 
$2\zeta \bar{N} < \kappa$. 
In the interval $ (\epsilon + \kappa) > 2\zeta \bar{N} \ge \kappa$, 
the state is nonclassical and separable. 
The steerability of the state is determined by 
\begin{eqnarray}
&&\!\!\!\!\!\!{\rm Spec}_{\mathbb R} (
\widetilde{\bf D}_{ [ i{\bf \Pi}_{1}] } )  =  
{\rm Spec}_{\mathbb R} (
\widetilde{\bf D}_{ [ i{\bf \Pi}_{2}] } )  =       \nonumber \\
&&\!\!\!\!\!\! \left\{ ( 2 \bar{N} +\tfrac{1}{2}) \zeta \pm \sqrt{\zeta^2 + \kappa^2}, 
        ( 2 \bar{N} +\tfrac{3}{2}) \zeta \pm \sqrt{\zeta^2 + \kappa^2}    \right\}.    
\end{eqnarray}
As a consequence of the chosen parameters, 
the steady state is symmetric with respect to the steerings of 
both partitions. This is state is steerable if and only if 
$( 2 \bar{N} +\tfrac{1}{2}) \zeta < \sqrt{\zeta^2 + \kappa^2} $. 

%%%%%%%%%%%%%%%%%%%%%%%%%%%%%%%%%%%%%%%%%%%%%%%%%%%%%%%%%%%%%%%%%%%%%%%%%%%%%%%%%%%%%%%%%
\section{Engineering Steady-States}\label{eps}               %%%%%%%%%%%%%%%%%%%%%%%%%%%%
%%%%%%%%%%%%%%%%%%%%%%%%%%%%%%%%%%%%%%%%%%%%%%%%%%%%%%%%%%%%%%%%%%%%%%%%%%%%%%%%%%%%%%%%%
The unavoidable influence of uncontrollable degrees of freedom are usually responsible 
for losses of the quantumness of a system through the procedure called decoherence.     
However, a steady-state with a desired quantum property can be produced by controlling 
the parameters of the system and of the environmental action.
As the examples of the last section, 
systems of bosonic degrees of freedom have been extensively studied in what concerns 
entanglement generation \cite{tan,houhou}, 
production of pure states \cite{koga,ikeda}, and 
engineering of graph states \cite{koga,houhou}. 
On the experimental side, 
realizations of these techniques in the context of atomic ensembles were already 
performed \cite{krauter}. 

To develop a simple theoretical engi\-nee\-ring-sta\-te method for 
bosonic degrees of freedom, we will use the results provided by the 
Williamson theorem \cite{williamson,gosson,simon}: 
%
%%%%%%%%%%%%%%%%%%%%%%%%%%%%%%%%%%%%%%%%%%%%%%%%%%%%%%%
\begin{theo}[Williamson]
Let 
$\mathbf M \in {\rm Mat}(2n,\mathbb R)$ 
be a positive definite matrix: 
$\mathbf M = \mathbf M^\top > 0$. 
This matrix can be diagonalized by 
a symplectic congruence, {\it i.e.}, there exists 
$\mathsf S \in {\rm Sp}(2n, \mathbb R)$ 
such that 
\begin{equation}                                                                         \label{tw1}      
\mathsf S \mathbf M \mathsf S^\top 
= 
{\rm Diag}(\mu_1,...,\mu_n,\mu_1,...,\mu_n) =:{\bf \Lambda }, 
\end{equation}
where
$\mu_j \ge \mu_k > 0 \,\,\, \text{for} \,\,\, j > k. $
\end{theo}
%%%%%%%%%%%%%%%%%%%%%%%%%%%%%%%%%%%%%%%%%%%%%%%%%%%%%%%
\noindent The double-paired ordered set (or the diagonal matrix) $\bf \Lambda$ 
is called {\it symplectic spectrum} of $\mathbf M$, and $\mu_k$ are 
its symplectic eigenvalues (\hypertarget{SE}{SE}). 
These can be found from the (Euclidean) eigenvalues of 
$\mathsf J \mathbf M$ \cite{gosson}, which turn out to be 
\begin{equation}                                                                         \label{tw2}
{\rm Spec_{\mathbb C}}(\mathsf J \mathbf M) = 
{\rm Diag}(i\mu_1, ...,i\mu_n,-i\mu_1,...,-i\mu_n).
\end{equation}
%
% The matrix $\mathsf S$ that diagonalizes $\bf M$ admits a suitable decomposition as
% %
% \begin{equation}                                                                         \label{tw3}
% \mathsf S = \sqrt{\bf \Lambda}  O \sqrt\mathbf{M}, \,\,\,\, O \in {\rm O}(2n).
% \end{equation}
% %
% From the symplectic condition $\mathsf S \mathsf J \mathsf S^\top = \mathsf J $, 
% one can see that $O$ must obey 
% %
% %\begin{equation}                                                                         \label{tw4}
% $O\sqrt{\bf M}\mathsf J \sqrt{\bf M} O^\top = 
% {\bf \Lambda }\mathsf J.$
% %\end{equation}

% %%%%%%%%%%%%%%%%%%%%%%%%%%%%%%%%%%%%%%%%%%%%%%%%%%%%%%%%%%%%%%%%%%%%%%%%%%%%%%%%%%%%%%%%%
% \subsection{General Procedure} %%%%%%%%%%%%%%%%%%%%%%%%%%%%%%%%%%%%%%%%%%%%%%%%%%%%%%%%%%
% %%%%%%%%%%%%%%%%%%%%%%%%%%%%%%%%%%%%%%%%%%%%%%%%%%%%%%%%%%%%%%%%%%%%%%%%%%%%%%%%%%%%%%%%%
Suppose one wants to design a reservoir structure able to produce a steady-state 
(with $n$ degrees of freedom) described by a $2n \times 2n$ \hy{CM}{CM} ${\bf V}'$. 
If one identifies this \hy{CM}{CM} with $\bf M$ in (\ref{tw1}), 
the first step is to find a suitable \hy{LE}{LE} 
able to produce its corresponding symplectic spectrum $\bf \Lambda$ as a solution, 
{\it i.e.}, its is necessary to find matrices ${\bf \Gamma}'$ and ${\bf D}'$ 
satisfying the LE $\lfloor {\bf \Lambda}, {\bf \Gamma}', {\bf D}' \rceil$.  
Assuming that it is possible to design a system-reservoir structure with this LE,  
one applies the symplectic covariance (\ref{simcov}) to it and finds 
\begin{equation}                                                                         \label{engcov}
\left\lfloor  {\bf V}' , 
         \mathsf S^{-1} {\bf \Gamma}'\mathsf S, 
         \mathsf S^{-1} {\bf D}'\mathsf S^{-\top} \right\rceil,   
\end{equation}
which, by Eq.(\ref{tw1}), is the \hy{LE}{LE} with solution 
${\bf V }' = \mathsf S^{-1} {\bf \Lambda }\mathsf S^{-\top}$.   
Reservoir engineering can then be realized by finding some convenient 
matrix $\bf \Lambda$ and a system-reservoir structure suitable to the unitary 
transformations, as in (\ref{unicov}). 
Following Eq.(\ref{simcov}), the engineered Hamiltonian and Lindblad operators 
will be, respectively, such that 
${\bf H}' = {\sf S}^{\!\top} \! {\bf H} {\sf S}$ and 
$\lambda' = {\sf S}^{-1} {\lambda}$.  

A peculiar example using (\ref{engcov}) appears when one is able to produce a 
reservoir structure such that 
${\bf D}' = \beta \bf \Lambda$ and 
${\bf \Gamma}' = - \beta \mathsf I_{2n}$ with $\beta > 0$, 
thus the use of (\ref{tw1}) gives that 
$\left\lfloor  {\bf V}', - \beta \mathsf I_{2n}, \beta {\bf V}' \right\rceil$. 
This shows that a Lindblad equation with  
${\bf \Upsilon} = \tfrac{1}{2}\beta {\bf V}' + i \beta \mathsf I_{2n}$ [see Eq.(\ref{upsilon})] 
will have a steady-state automatically given by a Gaussian with \hy{CM}{CM} ${\bf V}'$.

In fact, the matrix $\bf \Lambda$ in (\ref{tw1}) is like the 
\hy{CM}{CM} in (\ref{loccm}), 
thus any $\bf \Gamma'$ as in (\ref{locgamma}) and $\bf D'$ as in (\ref{locd}), 
which are invariants under the local rotations in (\ref{locrot}), 
are appropriate to the first step of the method. 
It is noticeable that not all reservoir structure are suitable to produce 
a diagonal matrix $\bf \Lambda$ as a \hy{CM}	{CM} of a steady-state, 
since it can not have the desired invariant structure for performing the 
first step.
However, it is still possible to use the covariance relation (\ref{simcov}) 
to design a specific steady-state from a known simple steady-state of some system, 
one of these cases (the OPO in Sec.~\ref{opo}) will be analyzed at the end of this section. 		

% %%%%%%%%%%%%%%%%%%%%%%%%%%%%%%%%%%%%%%%%%%%%%%%%%%%%%%%%%%%%%%%%%%%%%%%%%%%%%%%%%%%%%%%%%
% \subsection{A Particular Class of States}          %%%%%%%%%%%%%%%%%%%%%%%%%%%%%%%%%%%%%%
% %%%%%%%%%%%%%%%%%%%%%%%%%%%%%%%%%%%%%%%%%%%%%%%%%%%%%%%%%%%%%%%%%%%%%%%%%%%%%%%%%%%%%%%%%
The special condition presented in Corol.\ref{corol:diag} can be used to engineer a 
specific and important class of states.   
Explicitly, we will use the relation (\ref{simcov}) to generate a state with \hy{CM}{CM} 
${\bf V}' = \alpha \mathsf{SS}^{\!\top}$, where  
$\mathsf S \in {\rm Sp}(2n,\mathbb R)$ and $\alpha \ge 1$. 
In other words, we want to know which matrices 
$\bf \Gamma_{\rm p}$ \hy{AS}{AS} and 
${\bf D}_{\rm p} \ge 0$ 
are necessary to construct a \hy{LE}{LE} with the solution given 
by the desired \hy{CM}{CM}. 
Note that the above-mentioned \hy{CM}{CM} represents a pure state if and only if 
$\alpha = 1$ \cite{koga,simon,gosson}. 

As stated in Corol.\ref{corol:diag}, the $n$-mode Gibbs state with 
${\bf V} = \alpha \mathsf I_{2n}$ is generated by any \hy{LE}{LE} of the form 
$\lfloor {\bf V}, {\bf \Gamma}', -\alpha( {\bf \Gamma}' + {\bf \Gamma}'^{\!\top} ) \rceil$,  
for all matrices $\bf \Gamma'$ \hy{AS}{AS} such that 
$( {\bf \Gamma}' + {\bf \Gamma}'^{\!\top} ) \le 0$.  
Using the relation (\ref{simcov}), we are able to obtain a covariant \hy{LE}{LE} 
$ \lfloor {\bf V}', 
          {\bf \Gamma}_{\rm p}, {\bf D}_{\rm p} \rceil$ 
with 
\begin{equation}                                                                         \label{es1}
\!\!\!{\bf V}' = \alpha \mathsf{SS}^{\!\top}, \,\, 
{\bf \Gamma}_{\rm p} = \mathsf S {\bf \Gamma}' \mathsf S^{-1}, \,\, 
{\bf D}_{\rm p} = -\alpha \, \mathsf S ( {\bf \Gamma}' + 
                                         {\bf \Gamma}'^{\!\top} )\mathsf S^{\!\top}. 
\end{equation}
It is interesting to note that, recalling Corol.\ref{corol:diag}, 
the value of $\alpha$ only depends on the reservoir structure  
(it does not depend on the Hamiltonian of the system) used to 
prepare the initial state ${\bf V} = \alpha \mathsf I_{2n}$; 
conveniently the Hamiltonian at this stage can be taken as zero.   
Another remarkable fact is that any one-mode Gaussian state $(n = 1)$
is included in the scheme provided by Eq.(\ref{es1}): 
these states have a \hy{CM}{CM} written as ${\bf V}'$ in ({\ref{es1}})
with $\alpha = \mu_1 \ge 1$ and $\mathsf S \in {\rm Sp}(2,\mathbb R)$, 
as a consequence of the Williamson theorem. 

In \cite{koga}, the authors established specific conditions for a system to be driven to a 
pure steady-state when evolving under the \hy{LME}{LME} subjected 
to the restrictions in (\ref{ham-lind}). 
Equation (\ref{es1}) with $\alpha = 1$ constitutes a simple connection with some 
of their results. 
%
%%%%%%%%%%%%%%%%%%%%%%%%%%%%%%%%%%%%%%%%%%%%%%%%%%%%%%%%%%%%%%%%%%%%%%%%%%%%%%%%%%%%%%%%%
\subsection*{Examples}                                               %%%%%%%%%%%%%%%%%%%%
%%%%%%%%%%%%%%%%%%%%%%%%%%%%%%%%%%%%%%%%%%%%%%%%%%%%%%%%%%%%%%%%%%%%%%%%%%%%%%%%%%%%%%%%%
\paragraph{Two Mode Thermal Squeezed States (TMTSS).} 
Consider the following \hy{CM}{CM} 
\begin{equation}                                                                         \label{cm02}
{\bf V}' =  (2\bar n + 1) \mathsf{S}_r \mathsf{S}_r^{\!\top},
\end{equation}
with 
\begin{equation}                                                                         \label{ex1}
{\mathsf S_r} := \left[
\left(\! \begin{array}{cc}
       \cosh r & \sinh r\\
       \sinh r & \cosh r\\       
       \end{array} \! \right) \!  \oplus \!
\left( \! \begin{array}{cc}
       \cosh r& -\sinh r\\
       -\sinh r& \cosh r\\       
       \end{array} \!\! \right)  \right],      
\end{equation}
where $r \ge 0$ is the squeezing parameter and 
$\bar n \ge 0$ is the mean number of thermal photons of both modes.
Note that ${\mathsf S_r} \in {\rm Sp}(4,\mathbb R)$.
Following the criteria (\ref{bfsep}) and (\ref{gst}), 
the state in (\ref{cm02}) will be respectively 
entangled iff $r > {\rm ln}(2\bar n + 1)$ 
and steerable iff $r > {\rm cosh}^{-1}(2\bar n + 1)$.

If one wants to engineer a reservoir with the steady-state given in (\ref{cm02}), 
it is possible to apply the scheme in Eq.(\ref{es1}).  
For this purpose, one needs a reservoir structure able to produce 
a steady-state of the form ${\bf V} = \alpha \mathsf I_{2n}$ with 
$\alpha = (2\bar n + 1)$. 
An example of a system with this steady-state 
is the one in Sec.\ref{tor} with $\bar N_1 = \bar N_2$.
The system in Sec.\ref{to} can also be used to this end, 
but now, besides the condition $\bar N_1 = \bar N_2$, 
one needs to take $\kappa = 0$  --- this condition is necessary 
to guarantee that ${\bf \Gamma} + {\bf \Gamma}^{\!\top} \le 0$,
as required in Corol.\ref{corol:diag}. 
Recalling that the Hamiltonian dynamics does not affect the value 
of $\alpha$ in (\ref{es1}), one can use either  
$\omega_1 = \omega_2 = \kappa = 0$ and $\zeta_1 = \zeta_2 := \zeta$ 
in (\ref{dynsys}), or 
$\varpi_1 = \varpi_2 = \Omega = 0$ and $\zeta_1 = \zeta_2 := \zeta$ 
in (\ref{dynsys2}) to obtain 
a reservoir structure with 
${ \bf \Gamma }'  = -\frac{\zeta}{2}\mathsf I_{4}$ and 
${\bf D}' = (2\bar n + 1)\zeta \, \mathsf I_{4}$. 
This structure produces the steady-state 
${\bf V} =  (2\bar n + 1) \mathsf I_{2n}$. 

Now, one needs to apply the covariance relation (\ref{es1}) 
and determine the engineered reservoir with matrices 
${\bf \Gamma}_{\rm p}$ and ${\bf D}_{\rm p}$ through the symplectic matrix (\ref{ex1}), 
{\it i.e.},
\begin{equation}                                                                         \label{ex5}
{\bf \Gamma}_{\rm p} = -\frac{\zeta}{2}\mathsf I_{4}, \,\,\, 
{\bf D}_{\rm p} = (2\bar n + 1)\zeta \,\mathsf{S}_r \mathsf{S}_r^{\!\top}. 
\end{equation}
From Eq.(\ref{simcov}), one can see that 
$\lambda' = \mathsf{S}_r \lambda$, and the Lindblad operators 
in (\ref{ham-lind}) for $\lambda'$ are characteristics of 
a squeezed thermal bath. 
%%%%%%%%%%%%%%%%%%%%%%%%%%%%%%%%%%%%%%%%%%%%%%%%%%%%%%%%%%%%%%%%%%%%%%%%%%%%%%%%%%%%%%%%%
\paragraph{OPO Steady-States.}
The steady-state in (\ref{opocm}) is a pure steady-state iff 
$\epsilon_1 = - \epsilon_2$ \cite{koga}. 
Under these conditions and if we define 
$\epsilon := \sqrt{( \kappa + \epsilon_2)/(\kappa-\epsilon_2)}$, 
the CM of this steady-state is written as 
\begin{equation}                                                                         \label{ex6}       
{\bf V}' = \mathsf S_{\rm p} \mathsf S_{\rm p}^\top, \,\,\, 
{\mathsf S_{\rm p}} :=  
\left(\! \begin{array}{cc}
       \frac{1 + \epsilon}{2} & \frac{1 - \epsilon}{2} \\ 
       \frac{1 - \epsilon}{2} &  \frac{1 + \epsilon}{2}  
\end{array} \! \right)
\oplus
\left(\! \begin{array}{cc}
       \frac{1 + \epsilon}{2\epsilon} & \frac{\epsilon-1}{2\epsilon} \\ 
       \frac{\epsilon-1}{2\epsilon} &  \frac{1 + \epsilon}{2\epsilon} 
\end{array} \! \right),  
\end{equation}
which is entangled for any value of $\epsilon$.
Obviously, if one wants to produce this pure state as a steady-state of an OPO, 
the only step is to produce an OPO satisfying the mentioned conditions. 
On the other side, 
it is not possible to produce it by using the covariance rules in (\ref{es1}) for an OPO, 
since Corol.\ref{corol:diag} requires ${\bf \Gamma} + {\bf \Gamma}^\top \ge 0$, 
which is not the case for $\bf \Gamma$ in (\ref{opog}) 
with $\epsilon_1 = - \epsilon_2$. 

By the Willianson theorem, the symplectic spectrum of the \hy{CM}{CM} of 
the pure state in (\ref{ex6}) is the identity matrix \cite{footnote3}. 
If one can choose suitable values of the parameters in (\ref{opod}) and (\ref{opog}) 
such that ${\bf \Gamma}'$ and $\bf D'$ satisfy the LE 
$\lfloor \mathsf I_4, {\bf \Gamma}', {\bf D}' \rceil$, 
then Eq.(\ref{engcov}) can be applied. 
In fact, this happens when $\epsilon_1 = \epsilon_2 = 0$.  
Thus, preparing an OPO system such that this last condition holds, 
Eq.(\ref{engcov}) becomes 
\begin{equation}                                                                         \label{ex7}
\left\lfloor \mathsf S_{\rm p} \mathsf S_{\rm p}^\top  , 
         \mathsf S_{\rm p} {\bf \Gamma}'\mathsf S_{\rm p}^{-1}, 
         \mathsf S_{\rm p} {\bf D}'\mathsf S_{\rm p}^{\top} \right\rceil.    
\end{equation}
with ${\bf D}'$ as in (\ref{opod}) and 
\begin{equation}                                                                         \label{ex8}
{ \bf \Gamma }'  =  
         \left[\!\! 
                 \begin{array}{cc}
                  -\frac{\kappa}{2} & 0 \\
                  -\kappa           & -\frac{\kappa}{2}
                 \end{array}\!\!
         \right] \! \oplus \!
          \left[\!\! 
                  \begin{array}{cc}
                   -\tfrac{\kappa}{2}  & 0   \\
                  -\kappa   & -\tfrac{\kappa}{2}
                 \end{array}\!\!
         \right].  
\end{equation}
Note that the matrices $\bf \Gamma'$ and $\bf D'$ in this example do not have the 
invariant structure in (\ref{locrot}). 
Note also that the same pure state can be engineered by using thermal baths. 
The recipe for this case is just (\ref{ex5}) 
but replacing $\mathsf S_r$ by $\mathsf S_{\rm p}$ and using $\bar n = 0$.

An entangled and steerable (with respect to both partitions) 
mixed state %, respectively according to the criteria (\ref{bfsep}) and (\ref{gst}) 
can also be prepared by following the same recipe.
Suppose that one wants to create the bellow state as a steady-state 
of an OPO-covariant-LE:  
\begin{equation}
{\bf V}' = {\mathsf S}_{\rm p}^{-1} 
          {\bf \Lambda} 
          {\mathsf S}_{\rm p}^{-\top}, 
\end{equation}
with $\mathsf S_{\rm p}$ defined in (\ref{ex6}) and 
\begin{equation}
{\bf \Lambda} = {\rm Diag}\left[1,(1-\epsilon)^{-1},1,(1+\epsilon)^{-1} \right]. 
\end{equation}
This matrix is the solution for the LE 
$\lfloor {\bf \Lambda}, {\bf \Gamma}', {\bf D}' \rceil$ 
with $\bf D'$ in (\ref{opod}) and $\bf \Gamma'$ in (\ref{opog}) both with 
$\kappa = 1$, $\epsilon_1 = 0$ and $\epsilon_2 = \epsilon$. 
By the same recipe as before, the LE in (\ref{engcov}) has the above $\bf V'$ 
as solution if we replace $\mathsf S$ by $\mathsf S_{\rm p}$ and the mentioned 
matrices $\bf \Gamma'$ and $\bf D'$. 

\section{Final Remarks} \label{conc}            %%%%%%%%%%%%%%%%%%%%%%%%%%%%%%%%%%%%%%%%%
%%%%%%%%%%%%%%%%%%%%%%%%%%%%%%%%%%%%%%%%%%%%%%%%%%%%%%%%%%%%%%%%%%%%%%%%%%%%%%%%%%%%%%%%%
%%%%%%%%%%%%%%%%%%%%%%%%%%%%%%%%%%%%%%%%%%%%%%%%%%%%%%%%%%%%%%%%%%%%%%%%%%%%%%%%%%%%%%%%%
Symmetries and properties of the Lyapunov equation were used to classify the features of 
the steady state of a \hy{LME}{LME} with a quadratic Hamiltonian and linear 
Lindblad operators. 
The connection with the Lyapunov equation eases the characterization of the state, 
a task that is typically difficult when performed using the master equation directly. 

For Gaussian steady-states, we focused on known bona-fide relations for the 
covariance matrix of a state. Specifically, we considered conditions for the 
classicality, separability, and steerability of Gaussian states. 
We remark, however, that the extension for any other bona-fide relation is 
straightforward and can be performed following the lines presented here. 
For instance, we can refer to the characterization of tripartite 
entanglement given in Ref.~\cite{giedke}. 
We also analyze the consequences for the covariance matrix of a steady-state 
when a transformation symmetry of the Lyapunov equation is performed.  

We focused our examples on systems with one or two degrees of freedom,
which has enabled us to compare the results of our corollaries with 
the results extracted directly from the covariance matrix of the system after 
solving the Lyapunov equation.
However, our results are applicable to systems with a generic number of 
degrees of freedom. For large systems, in particular in the absence of symmetries, 
numerical solutions might be needed to find the covariance matrix of the 
steady-state; for instance, the systems considered in Ref. \cite{nicacio3}. 
In this situation, instabilities associated with the algorithms for solving 
Lyapunov equations may arise \cite{hammarling}. 
The robustness of the analytical results shows the advantage with respect to 
either perturbations of the systems parameters or preparation imprecisions. 
In other words, our results are advantageous since one does not need to solve 
a Lyapunov equation to know some of the system properties or symmetries.
In addition, our results are suitable for the engineering of 
(a reservoir leading to a specific) steady-state of a \hy{LME}{LME}
having suitable properties and symmetries. 
%
%%%%%%%%%%%%%%%%%%%%%%%%%%%%%%%%%%%%%%%%%%%%%%%%%%%%%%%%%%%%%%%%%%%%%%%%%%%%%%%%%%%%%%%%%
%%%%%%%%%%%%%%%%%%%%%%%%%%%%%%%%%%%%%%%%%%%%%%%%%%%%%%%%%%%%%%%%%%%%%%%%%%%%%%%%%%%%%%%%%
\section*{Appendices} \appendix
%%%%%%%%%%%%%%%%%%%%%%%%%%%%%%%%%%%%%%%%%%%%%%%%%%%%%%%%%%%%%%%%%%%%%%%%%%%%%%%%%%%%%%%%%
%\vspace{-0.2cm}
%%%%%%%%%%%%%%%%%%%%%%%%%%%%%%%%%%%%%%%%%%%%%%%%%%%%%%%%%%%%%%%%%%%%%%%%%%%%%%%%%%%%%%%%%
\setcounter{equation}{0}
\renewcommand{\theequation}{A-\arabic{equation}}  
%%%%%%%%%%%%%%%%%%%%%%%%%%%%%%%%%%%%%%%%%%%%%%%%%%%%%%%%%%%%%%%%%%%%%%%%%%%%%%%%%%%%%%%%%
\hypertarget{Appendix}{\subsection*{Appendix I: Notations and Definitions}} \appendix  %%  
%%%%%%%%%%%%%%%%%%%%%%%%%%%%%%%%%%%%%%%%%%%%%%%%%%%%%%%%%%%%%%%%%%%%%%%%%%%%%%%%%%%%%%%%%
\vspace{-0.2cm}
Throughout the text we use some mathematical objects whose notations are defined here.
\begin{itemize}\itemsep1.5pt
%\item $ A := B$ \, : Definition, $A$ is defined by $B$, the same as $ B =: A$. 
%
\item[\textbullet]${\rm Mat}(m,\mathbb K)$: 
                  set of all $m \times m$ square matrices over the field $\mathbb K$.  
\item[\textbullet] ${\rm GL}(m, \mathbb K) := 
                   \{ {\bf M} \in {\rm Mat}(m,\mathbb K) | \det {\bf M} \ne 0\}  $: 
                   Ge\-ne\-ral linear group over field $\mathbb K$. 
\item[\textbullet] ${\rm Sp}(2m, \mathbb R):= 
                   \{ {\mathsf M} \in {\rm Mat}(2m,\mathbb R) | 
                   \mathsf M \mathsf J \mathsf M^\top = \mathsf J \}  $: 
                   Real symplectic group. 
\item[\textbullet] ${\rm O}(m) := \{ M \in {\rm Mat}(m,\mathbb R) | 
                    M M^\top = \mathsf I_{m} \}  $: Real orthogonal group.                   
\item[\textbullet] $\mathsf I_{m}$: Identity matrix in ${\rm Mat}(m,\mathbb K)$. 
\item[\textbullet] ${\bf 0}_m$: Zero matrix in ${\rm Mat}(m,\mathbb K)$.
\item[\textbullet] ${\rm Spec}_{\mathbb K}({\bf M}) := \{\nu_1,...,\nu_l\}$
                   is the spectrum of ${\bf M}\in{\rm Mat}(m,\mathbb K)$.  
                   It is the set of its eigenvalues  
                   $\nu_k \in \mathbb K, \, \forall k$ and $l \le m$. 
\item[\textbullet] ${\bf M}^{\!\top}$: Transposition of $\bf M$; \, 
                   ${\bf M}^{-\top}$:  Inverse of $\bf M^{\!\top}$. 
\item[\textbullet] ${\bf M}^{\ast}$: Complex conjugation of the elements of $\bf M$.  
\item[\textbullet] ${\rm In}({\bf M}) := (n_+,n_0,n_-)(\bf M) $: 
                   Inertia index, {\it i.e.},
                   the triple containing the number of eigenvalues of ${\bf M}$ 
                   with positive ($n_+$), null ($n_0$) and negative ($n_-$) real part. 
                   {\it N.B.} if ${\bf M}\in {\rm Mat}(m,\mathbb K) $, 
                   then $m = n_+ + n_0 + n_-$. 
\end{itemize}
In what follows, ${\bf M}\in {\rm Mat}(m,\mathbb K) $ and ${\bf M} = {\bf M}^\dagger$: 
\begin{itemize}
\item[\textbullet] ${\bf M} > 0 $ (resp. $ {\bf M}< 0$): 
                   Positive (resp. negative) definiteness of ${\bf M}$, 
                   {\it i.e.}, all its eigenvalues are positive (resp. negative). 
\item[\textbullet] ${\bf M} \ge 0$ (resp. $ {\bf M} \le 0$): Positive (resp. negative) 
                   semidefiniteness of ${\bf M}$, {\it i.e.}, 
                   all its eigenvalues are non-negative (resp. nonpositive). 
                   In this text, the statement ${\bf M} \ge 0$ (resp. ${\bf M} \le 0$) 
                   means that $\bf M$ can, but not necessarily, 
                   have null eigenvalues. This is the same as say that the set of 
                   matrices such that ${\bf M}>0$ (resp. ${\bf M} < 0$)
                   is a subset of the ones satisfying 
                   ${\bf M}\ge 0$ (resp. ${\bf M} \le 0$). 
\end{itemize} 

It is noteworthy that, following our definitions, 
the sum of two positive (resp. negative) semidefinite matrices is 
positive (resp. negative) semidefinite, 
{\it i.e.}, the sum will have non-negative (resp. non-positive) eigenvalues. 
In addition, the sum of two positive (resp. negative) definite matrices is 
positive (resp. negative) definite.
\vspace{-0.3cm}
%%%%%%%%%%%%%%%%%%%%%%%%%%%%%%%%%%%%%%%%%%%%%%%%%%%%%%%%%%%%%%%%%%%%%%%%%%%%%%%%%%%%%%%%%
%%%%%%%%%%%%%%%%%%%%%%%%%%%%%%%%%%%%%%%%%%%%%%%%%%%%%%%%%%%%%%%%%%%%%%%%%%%%%%%%%%%%%%%%%
\hypertarget{Appendix2}{\subsection*{Appendix II: On the P-Representability of States}}%% 
\appendix %%%%%%%%%%%%%%%%%%%%%%%%%%%%%%%%%%%%%%%%%%%%%%%%%%%%%%%%%%%%%%%%%%%%%%%%%%%%%%% 
%%%%%%%%%%%%%%%%%%%%%%%%%%%%%%%%%%%%%%%%%%%%%%%%%%%%%%%%%%%%%%%%%%%%%%%%%%%%%%%%%%%%%%%%%
\vspace{-0.3cm}
Due to the absence of a proof in the literature, this appendix is devoted to prove 
that the necessary and sufficient condition for P-Representability of a 
$n$-mode Gaussian state is Eq.(\ref{bfclas}).

A quantum state $\hat \rho$ is P-representable, by definition, 
if it can be written as a convex and regular sum of coherent 
states through the Glauber-Sudarshan $P$-function \cite{sudarshan}: 
\begin{equation}                                                                         \label{prep}
\hat \rho = \int \! P(\zeta) \, |\zeta \rangle \! \langle \zeta | \, d^{2n}\zeta, 
\end{equation}
where $\zeta \in \mathbb R^{2n}$ and $|\zeta \rangle$ is a coherent state. 

The sufficient condition is proved in \cite{englert} for two mode Gaussian states, 
{\it i.e.}, $n = 2$ in (\ref{prep}). 
The extension for any $n$-mode state (not only the Gaussians) 
follows the same recipe: using the definition of the \hy{CM}{CM} (\ref{cmdef}) with the 
$\hat \rho$ in (\ref{prep}), Eq.(\ref{bfclas}) follows immediately.

To prove the necessary condition (only for Gaussian states), 
we choose two Gaussian states $\hat \rho$ and $\hat \rho_0$, 
with the respective \hy{CM}{CMs} 
${\bf V}$ and ${\bf V}_0$ such that ${\bf V} \ge {\bf V}_0$. 
These states can be related through a Gaussian noise channel \cite{caves}:    
\begin{equation}                                                                         \label{bosonico}
\hat \rho = \frac{1}{(\pi\hbar)^n} \int^{+ \infty}_{-\infty} 
\frac{ {\rm e}^{-\frac{1}{\hbar} \zeta \cdot \Delta^{-1} \zeta} } 
     { \sqrt{\rm Det \Delta} } \,\, 
\hat T_\zeta \hat\rho_0 \hat T_\zeta^\dag \,\,\, d^{2n} \zeta.                                               
\end{equation}
The operators $\hat T_\zeta$ are the Weyl displacement operators \cite{gosson}, 
and $\Delta := {\bf V} - {\bf V}_0 \ge 0$. 
Mathematically speaking, Eq.(\ref{bosonico}) express the very known fact that 
the convolution of two Gaussian functions is a Gaussian function. 
If we choose ${\bf V}_0 = \mathsf I_{2n}$ and 
$\hat \rho_0$ as a vacuum state, the positive-semidefiniteness of $\Delta$
implies relation (\ref{bfclas}), as we should prove. 

\vspace{-0.2cm}
%%%%%%%%%%%%%%%%%%%%%%%%%%%%%%%%%%%%%%%%%%%%%%%%%%%%%%%%%%%%%%%%%%%%%%%%%%%%%%%%%%%%%%%%%
%%%%%%%%%%%%%%%%%%%%%%%%%%%%%%%%%%%%%%%%%%%%%%%%%%%%%%%%%%%%%%%%%%%%%%%%%%%%%%%%%%%%%%%%%
\acknowledgments         %%%%%%%%%%%%%%%%%%%%%%%%%%%%%%%%%%%%%%%%%%%%%%%%%%%%%%%%%%%%%%%%
%%%%%%%%%%%%%%%%%%%%%%%%%%%%%%%%%%%%%%%%%%%%%%%%%%%%%%%%%%%%%%%%%%%%%%%%%%%%%%%%%%%%%%%%%
%%%%%%%%%%%%%%%%%%%%%%%%%%%%%%%%%%%%%%%%%%%%%%%%%%%%%%%%%%%%%%%%%%%%%%%%%%%%%%%%%%%%%%%%%
FN acknowledges the warm hospitality of the CTAMOP at Queen's University Belfast. 
Insightful discussions with A. Xuereb at the beginning of the work and with F. Semi\~ao 
through the writing of the whole work were valuable.
FN and MP are supported by the CNPq ``Ci\^{e}ncia sem Fronteiras'' 
programme through the ``Pesquisador Visitante Especial'' initiative 
(Grant No. 401265/2012-9).
MP acknowledges financial support from the UK EPSRC (EP/G004579/1). 
MP and AF are supported by 
the John Templeton Foundation (grant ID 43467), and the EU Collaborative Project TherMiQ 
(Grant Agreement 618074). 
MP gratefully acknowledge support from  the COST Action MP1209 
``Thermodynamics in the quantum regime".
%
%%%%%%%%%%%%%%%%%%%%%%%%%%%%%%%%%%%%%%%%%%%%%%%%%%%%%%%%%%%%%%%%%%%%%%%%%%%%%%%%%%%%%%%%%

%%%%%%%%%%%%%%%%%%%%%%%%%%%%%%%%%%%%%%%%%%%%%%%%%%%%%%%%%%%%%%%%%%%%%%%%%%%%%%%%%%%%%%%%%

\begin{thebibliography}{99}
\bibitem{albert2}
A good and extensive set of references 
for applications can be found in Ref.\cite{albert}.
\bibitem{footnote2}
As stated in Ref.\cite{nicacio2}, 
if we consider the dynamics of Hamiltonian (\ref{hess1}), 
move to an interaction picture with respect to the free oscillators 
evolutions with frequencies $\omega_1$ and $\omega_2$, 
perform the RWA (discarding the fast oscillation terms) 
and back again to Schr\"odinger picture, we obtain Eq.(\ref{hess2}) with 
$ \varpi_i = \omega_i + \tfrac{\kappa}{4}$ and $ \Omega = \tfrac{\kappa}{4} $.
\end{thebibliography}
\end{document}